\documentclass[conference]{IEEEtran}
\usepackage[dvips]{graphicx}
\usepackage{subfigure}
\usepackage{mdwlist}
\usepackage{amsmath}
\usepackage{amsthm}
\usepackage[linesnumbered,ruled]{algorithm2e}
\usepackage{cite}
\usepackage{amssymb}
\usepackage{multirow}
\usepackage{makecell}
\usepackage{enumerate}
\usepackage{tabularx}

\newtheorem{definition}{Definition}

\newtheorem{theorem}{Theorem}
\newtheorem{lemma}{Lemma}
\newtheorem{corollary}{Corollary}

\newtheorem{proposition}{Proposition}
\newtheorem{example}{Example}

\begin{document}

\title{OMG: How Much Should I Pay Bob in Truthful Online Mobile Crowdsourced Sensing?}
\author{\IEEEauthorblockN{Dong Zhao$^{\dag, \ddag}$, Xiang-Yang Li$^\ddag$, and Huadong Ma$^\dag$}
\IEEEauthorblockA{$\dag$ Beijing Key Laboratory of Intelligent Telecommunications Software and Multimedia,\\Beijing University of Posts and Telecommunications, China\\}
\IEEEauthorblockA{$\ddag$ Department of Computer Science, Illinois Institute of Technology, Chicago, IL, USA\\
{Email: dzhao@bupt.edu.cn; xli@cs.iit.edu; mhd@bupt.edu.cn}}
}
\maketitle

\begin{abstract}
Mobile crowdsourced sensing (MCS) is a new paradigm which takes advantage of the pervasive smartphones to efficiently collect data, enabling numerous novel applications.
To achieve good service quality for a MCS application, incentive mechanisms are necessary to attract more user participation.
Most of existing mechanisms apply only for the \emph{offline} scenario where all users' information are known a priori.
On the contrary, we focus on a more real scenario where users arrive one by one \emph{online} in a random order.
We model the problem as an \emph{online auction} in which the users submit their private types to the crowdsourcer over time, and the crowdsourcer aims to select a subset of users before a specified deadline for maximizing the total value of the services provided by selected users under a budget constraint.
We design two \emph{online mechanisms}, \emph{OMZ} and \emph{OMG}, satisfying the \emph{computational efficiency}, \emph{individual rationality}, \emph{budget feasibility}, \emph{truthfulness}, \emph{consumer sovereignty} and \emph{constant competitiveness} under the zero arrival-departure interval case and a more general case, respectively.
Through extensive simulations, we evaluate the performance and validate the theoretical properties of our online mechanisms.
\end{abstract}

\section{Introduction}
\label{sec:introduction}
Crowdsourcing is a distributed problem-solving model in which a crowd of undefined size is engaged to solve a complex problem through an open call \cite{chatzimilioudis2012crowdsourcing}.
Nowadays, the proliferation of smartphones provides a new opportunity for extending existing web-based crowdsourcing applications to a larger contributing crowd, making contribution easier and omnipresent.
Furthermore, today's smartphones are programmable and come with a rich set of cheap powerful embedded sensors, such as GPS, WiFi/3G/4G interfaces, accelerometer, digital compass, gyroscope, microphone, and camera.
The great potential of the mobile phone sensing offers a variety of novel, efficient ways to opportunistically collect data, enabling numerous \emph{mobile crowdsourced sensing} (MCS) applications, such as Sensorly \cite{website:Sensorly} for constructing cellular/WiFi network coverage maps, SignalGuru \cite{koukoumidis2011signalguru}, Nericell \cite{mohan2008nericell} and VTrack \cite{thiagarajan2009vtrack} for providing traffic information, Ear-Phone \cite{rana2010ear} and NoiseTube \cite{stevens2010crowdsourcing} for making noise maps. For more details on MCS applications, we refer interested readers to several survey papers \cite{lane2010survey,ganti2011mobile,chatzimilioudis2012crowdsourcing}.

Adequate user participation is one of the most critical factors determining whether a MCS application can achieve good service quality.
Most of the current MCS applications \cite{website:Sensorly,koukoumidis2011signalguru,mohan2008nericell,thiagarajan2009vtrack,rana2010ear,stevens2010crowdsourcing} are based on voluntary participation.
While participating in a MCS campaign, smartphone users consume their own resources such as battery and computing power, and expose their locations with potential privacy threats.
Thus, incentive mechanisms are necessary to provide participants with enough rewards for their participation costs.
At present, only a handful of work \cite{danezis2005much,lee2010sell,duan2012incentive,yang2012crowdsourcing,jaimes2012location} focuses on incentive mechanism design for MCS applications.
All of these work applies only for the \emph{offline} scenario in which all of participating users report their types, including the tasks they can complete and the bids, to the crowdsourcer (campaign organizer) in advance, and then the crowdsourcer selects a subset of users after collecting the information of all users to maximize its utility (e.g., the total value of all tasks that can be completed by selected users).

In practice, however, users always arrive one by one \emph{online} in a random order and user availability changes over time.
Therefore, an \emph{online incentive mechanism} is necessary to make irrevocable decisions on whether to accept a user's task and bid, based solely on the information of users arriving before the present moment, without knowing future information.

In this paper we consider a general problem: the crowdsourcer aims to select a subset of users before a specified deadline, so that the total value of the services provided by selected users is maximized under the condition that the total payment to these users does not exceed a budget constraint.
Specially, we investigate the case where the value function of selected users is monotone submodular.
This case can be applied in many real scenarios.
For example, many MCS applications \cite{website:Sensorly,koukoumidis2011signalguru,mohan2008nericell,thiagarajan2009vtrack,rana2010ear,stevens2010crowdsourcing} aim to select users to collect sensing data so that the roads in a given region can be covered before a specified deadline, where the coverage function is typically monotone submodular.
In addition, the cost and arrival/departure time of each user are private and only known to itself.
We consider users who are game-theoretic and seek to make strategy (possibly report an untruthful cost or arrival/departure time) to maximize their individual utility in equilibrium.
Therefore, the problem can be modeled as an \emph{online auction}, for which we can design the online mechanism based on the theoretical foundations of mechanism design and online algorithms.

Our objective is to design online mechanisms satisfying six desirable properties: \emph{computational efficiency}, \emph{individual rationality}, \emph{budget feasibility}, \emph{truthfulness}, \emph{consumer sovereignty} and \emph{constant competitiveness}.
Informally, \emph{computational efficiency} ensures the mechanism can run in real time, \emph{individual rationality} ensures each participating user has a non-negative utility, \emph{budget feasibility} ensures the crowdsourcer's budget constraint is not violated, \emph{truthfulness} ensures the participating users report their true costs (\emph{cost-truthfulness}) and arrival/departure times (\emph{time-truthfulness}), \emph{consumer sovereignty} ensures each participating user has a chance to win the auction, and \emph{constant competitiveness} guarantees that the mechanism performs close to the optimal solution in the offline scenario where all the information of all users are known to the crowdsourcer in advance.

The main idea behind our online mechanism is to adopt a multiple-stage sampling-accepting process.
At every stage the mechanism allocates tasks to a smartphone user only if its marginal density is not less than a certain density threshold that has been computed using previous users' information, and the budget allocated for the current stage has not been exhausted.
Meanwhile, the user obtains a payment equaling to the ratio of its marginal value to the density threshold.
The density threshold is computed in a manner that guarantees desirable performance properties of the mechanism.
We firstly consider the \emph{zero arrival-departure interval} case where the arrival time of each user equals to its departure time (Section \ref{sec:special case}).
In this case, achieving time-truthfulness is trivial.
We present an online mechanism \emph{OMZ} satisfying all desirable properties under this special case without considering the time-truthfulness.
Then we revise the \emph{OMZ} mechanism, and present another online mechanism \emph{OMG} satisfying all desirable properties under the \emph{general} case (Section \ref{sec:general case}).

The remainder of this paper is organized as follows.
In Section \ref{sec:problem formulation} we describe the MCS system model, and formulate the problem as an online auction.
We then present two online mechanisms, \emph{OMZ} and \emph{OMG}, satisfying all desirable properties under the \emph{zero arrival-departure interval} case and the \emph{general} case in Section \ref{sec:special case} and \ref{sec:general case}, respectively.
Performance evaluations are presented in Section \ref{sec:performance evaluation}.
We review the related work in Section \ref{sec:related work}, and conclude this paper in Section \ref{sec:conclusion}.

\section{System Model and Problem Formulation}
\label{sec:problem formulation}
We use Fig. \ref{fig-systemModel} to illustrate a MCS system.
The system consists of a \emph{crowdsourcer}, which resides in the cloud and consists of multiple sensing servers, and many smartphone \emph{users}, which are connected to the cloud by cellular networks (e.g., GSM/3G/4G) or Wi-Fi connections.
The crowdsourcer first publicizes a MCS campaign in a \emph{Region of Interest (RoI)}, aiming to finding some users to complete a set of tasks $\Gamma=\{\tau_1,\tau_2,\ldots,\tau_m\}$ in the RoI before a specified deadline $T$.
Assume that a crowd of smartphone users $\mathcal{U}=\{1,2,\ldots,n\}$ interested in participating in the crowdsourcing campaign arrive online in a random order, where $n$ is unknown.
Each user $i$ has an arrival time $a_i\in \{1, \ldots, T\}$, a departure time $d_i\in \{1, \ldots, T\}$, $d_i\geq a_i$, and a subset of tasks $\Gamma_i\subseteq \Gamma$ it can complete within this time interval according to its willness and ability.
Meanwhile, user $i$ also has an associated cost $c_i \in \mathbb{R}_+$ for performing sensing tasks according to its current state such as the residual battery energy of the smartphone and its willness.
All information constitutes the \emph{type} of user $i$, $\theta_i=(a_i,d_i,\Gamma_i,c_i)$.
In this paper we consider two models with respect to the distribution of users:
\begin{itemize*}
\item \textbf{The i.i.d. model:} The costs and values of users are i.i.d. sampled from some \emph{unknown} distributions.
\item \textbf{The secretary model:} An adversary gets to decide on the costs and values of users, but not on the \emph{order} in which they are presented to the crowdsourcer.
\end{itemize*}
In fact, the i.i.d. model is a special case of the secretary model, since the sequence can be determined by first picking a multi-set of costs or values from the (unknown) distribution, and then permuting them randomly.
Note that these two models are different from the \emph{oblivious adversarial model}, where an adversary chooses a \emph{worst-case} input stream including the users' costs, values and their arrival orders.

We model the interactive process between the crowdsourcer and users as an \emph{online auction}.
Each user expects a \emph{payment} in return for its service.
Therefore, it makes a \emph{reserve price}, called \emph{bid}, for selling its sensing data.
When a user arrives, the crowdsourcer must decide whether to buy the service of this user, and if so, at what price, before it departs.
Assume that the crowdsourcer has a budget constraint $B$ indicating the maximum value that it is willing to pay.
Therefore, the crowdsourcer always expects to obtain the maximum value from the selected users' services under the budget constraint.
\begin{figure}[!t]
\centering{
\includegraphics[width=3.5in]{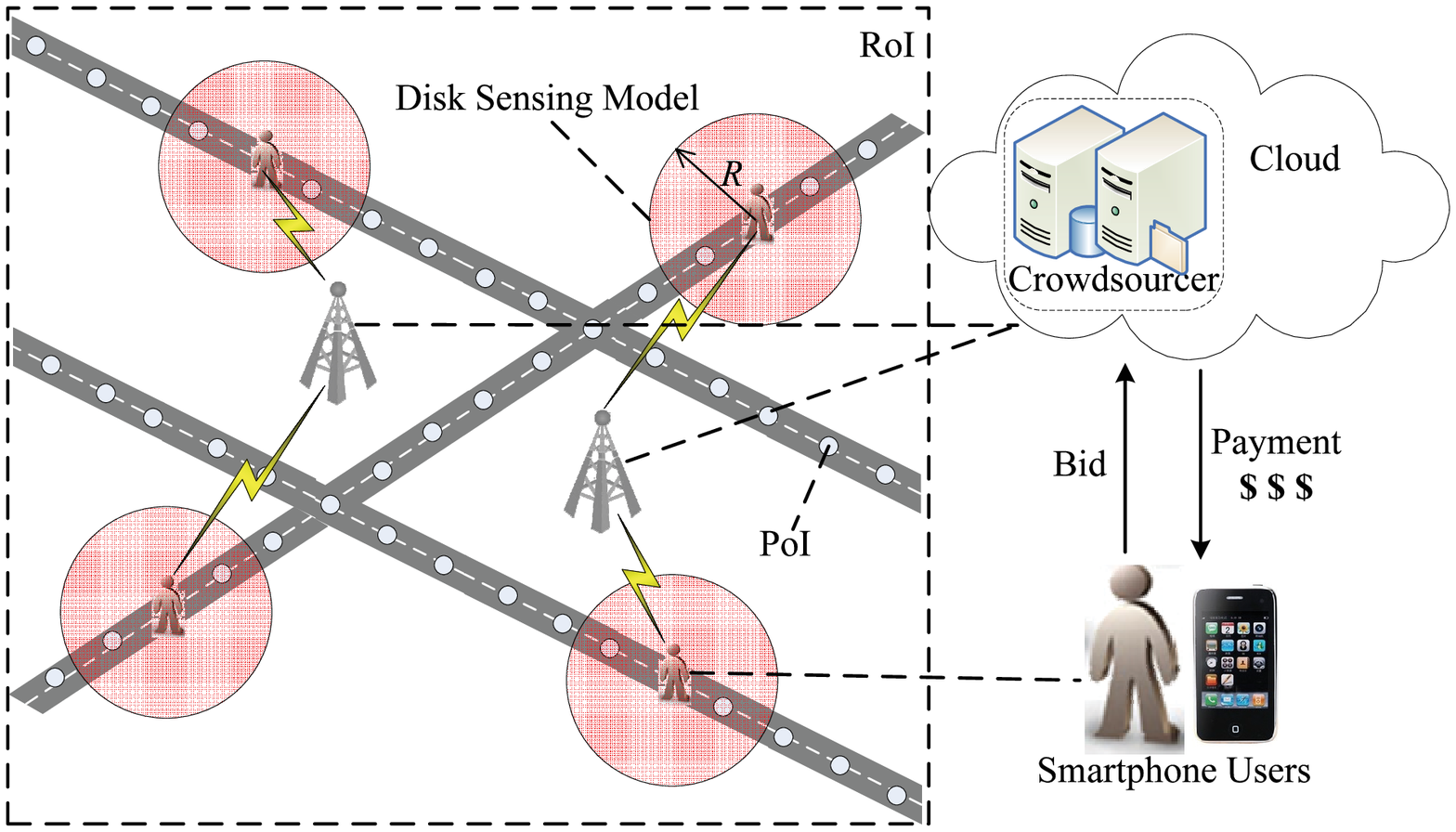}}
\caption{Illustration of a mobile crowdsourced sensing system.}
\label{fig-systemModel}
\vspace{-10pt}
\end{figure}

In the online auction we consider users that are game-theoretic and seek to make \emph{strategy} to maximize their individual utility
in equilibrium.
Note that the arrival time, departure time, and cost of user $i$ are private and only known to itself.
Only the task set $\Gamma_i$ must be true since the crowdsourcer can identify whether the announced tasks are performed.
In other words, user $i$ may misreport all information about its type except for $\Gamma_i$.
The budget and value function of the crowdsourcer are common knowledge.
Although our auctions do not require a user to declare its departure time until the moment of its departure, we find
it convenient to analyze our auctions as direct-revelation mechanisms (DRMs).
The strategyspace in an online DRM allows a user to declare some possibly untruthful type $\hat{\theta_i}=(\hat{a_i},\hat{d_i},\Gamma_i,b_i)$, subject to $a_i \leq \hat{a_i} \leq \hat{d_i} \leq d_i$.
Note that we assume that a user cannot announce an earlier arrival time or a later departure time than its true arrival/departure time.
In order to obtain the required service, the crowdsourcer needs to design an \emph{online mechanism} $\mathcal{M}=(f,p)$ consisting of an \emph{allocation} function $f$ and a \emph{payment} function $p$.
For any \emph{strategy sequence} $\hat{\theta}=(\hat{\theta_1},\ldots,\hat{\theta_n})$, the allocation function $f(\hat{\theta})$ computes an allocation of tasks for a selected subset of users $\mathcal{S}\in \mathcal{U}$, and the payment function $p(\hat{\theta})$ returns a vector $(p_1(\hat{\theta}),\ldots,p_n(\hat{\theta}))$ of payments to the users.
Note that, the crowdsourcer, when presented with the strategy $\hat{\theta_i}$ of user $i$, must decide whether to buy the service of user $i$, and if so, at what price before it departs.

The \emph{utility} of user $i$ is
\[
u_i=
\begin{cases}p_i-c_i,\quad \ \ & \mbox{if}\quad i\in \mathcal{S} ;\\
  0, \quad \ \ & \mbox{otherwise}.
\end{cases}
\]
Let $V(\mathcal{S})$ denote the \emph{value} function of the crowdsourcer over the selected subset of users $\mathcal{S}$.
The crowdsourcer expects to obtain the maximum value from the selected users' services under the budget constraint, i.e.,
    \begin{equation}\textbf{Maximize } V(\mathcal{S}) \textbf{ subject to } \sum_{i\in \mathcal{S}}p_i \leq B.\nonumber\end{equation}
In this paper, we focus on the case where $V(\mathcal{S})$ is monotone submodular. This case can be applied in many real scenarios.
\begin{definition}[Monotone Submodular Function]
\label{def:submodular}
Let $\Omega$ be a finite set. For any $X\subseteq Y \subseteq \Omega$ and $x\in \Omega\backslash Y$, a function $f: 2^\Omega \mapsto \mathbb{R}$ is called submodular if and only if
\[f(X\cup\{x\})-f(X)\geq f(Y\cup\{x\})-f(Y),\]
and it is monotone (increasing) if and only if $f(X)\leq f(Y)$, where $2^\Omega$ denotes the power set of $\Omega$, and $\mathbb{R}$ denotes the set of reals.
\end{definition}

\underline{An Application Example:} As illustrated in Fig. \ref{fig-systemModel}, we consider the scenario where the crowdsourcer expects to obtain the sensing data covering all roads in a RoI.
For convenience of calculations, we divide each road in the RoI into multiple discrete \emph{Points of Interest (PoIs)}, and the objective of the crowdsourcer is equivalent to obtaining the sensing data covering all PoIs before $T$.
The set of PoIs is denoted by $\Gamma=\{\tau_1,\tau_2,\ldots,\tau_m\}$.
Assume that each sensor follows a geometric disk sensing model with sensing range $R$, which means if user $i$ senses at a location $L_i$ and obtain a reading, then any PoI within the disk with the origin at $L_i$ and a radius of $R$ has been covered once.
The set of PoIs covered by user $i$ is denoted by $\Gamma_i\subseteq \Gamma$, which means the sensing tasks that user $i$ can complete.
Without loss of generality, assume that each PoI $\tau_j$ has a coverage requirement $r_j \in \mathbb{Z}_+$ indicating how many times it requires to be sensed at most.
The \emph{value} of the selected users to the crowdsourcer is:
\[
V(\mathcal{S})=\sum_{j=1}^m{\min\{r_j,\sum_{i\in \mathcal{S}}v_{i,j}\}},
\]
where $v_{i,j}$ equals to 1 if $\tau_j\in \Gamma_i$, and 0 otherwise.
\begin{lemma}
\label{lemma_valueFunction}
The value function $V(\mathcal{S})$ is monotone submodular.
\end{lemma}

The proof of Lemma \ref{lemma_valueFunction} is given in Appendix A.

Our objective is to design an online mechanism satisfying the following six desirable properties:
\begin{itemize*}
\item \textbf{Computational Efficiency:} A mechanism is \emph{computationally efficient} if both the allocation and payment can be computed in polynomial time as each user arrives.
\item \textbf{Individual Rationality:} Each participating user will have a non-negative utility: $u_i\geq 0$.
\item \textbf{Budget Feasibility:} We require the mechanism to be \emph{budget feasible}: $\sum_{i\in \mathcal{S}}p_i \leq B$.
\item \textbf{Truthfulness:} A mechanism is \emph{cost-} and \emph{time-truthful} (or simply called \emph{truthful}, or \emph{incentive compatible} or \emph{strategyproof}) if reporting the true cost and arrival/departure time is a \emph{dominant strategy} for all users. In other words, no user can improve its utility by submitting a false cost, or arrival/departure time, no matter what others submit.
\item \textbf{Consumer Sovereignty:} The mechanism cannot arbitrarily exclude a user; the user will be selected by the crowdsourcer and obtain a payment if only its bid is sufficiently low while others are fixed.
\item \textbf{Competitiveness:} The goal of the mechanism is to maximize the value of the crowdsourcer. To quantify the performance of the mechanism we compare its solution with the \emph{optimal solution}: the solution obtainable in the offline scenario where the crowdsourcer has full knowledge about users' types. A mechanism is $O(g(n))$-\emph{competitive} if the ratio between the online solution and the optimal solution is $O(g(n))$. Ideally, we would like our mechanism to be $O$(1)-\emph{competitive}.
\end{itemize*}

The importance of the first three properties is obvious, because they together guarantee that the mechanism can be implemented in real time and satisfy the basic requirements of both the crowdsourcer and users.
In addition, the last three properties are indispensable for guaranteeing that the mechanism has high performance and robustness.
The truthfulness aims to eliminate the fear of market manipulation and the overhead of strategizing over others for the participating users.
The consumer sovereignty aims to guarantee that each participating user has a chance to win the auction and obtain a payment, otherwise it will hinder the users' completion or even result in task starvation.
Besides, if some users are guaranteed not to win the auction, then being truthful or not will have the same outcome.
For this reason, the property satisfying both the consumer sovereignty and the truthfulness is also called \emph{strong truthfulness} by Hajiaghayi et al. \cite{hajiaghayi2004adaptive}.
Later we will show that satisfying consumer sovereignty is not trivial in the online scenario, which is in contrast to the offline scenario.
Finally, we expect that our mechanism has a constant competitiveness under both the \emph{i.i.d.} model and the \emph{secretary} model.
Note that no constant-competitive auction is possible under the \emph{oblivious adversarial} model \cite{bar2002incentive}.

Table \ref{table_notations} lists frequently used notations.
\begin{table}[t]
\begin{center}
\caption{Frequently used notations.}
\label{table_notations}
\begin{tabularx}{0.48\textwidth}{c|XcX}
  \hline
  Notation & Description\\
  \hline
  $\mathcal{U},n,i$ & set of users, number of users, and one user\\
  $\Gamma, m, \tau_j$ & set of tasks, number of tasks, and one task\\
  $B,B'$ & budget constraint and stage-budget\\
  $T,T',t$ & deadline, end time step of each stage, and each time step\\
  $a_i,\hat{a_i}$ & true arrival time and strategic arrival time of user $i$\\
  $d_i,\hat{d_i}$ & true departure time and strategic departure time of user $i$\\
  $\Gamma_i$ & set of user $i$'s tasks\\
  $c_i,b_i$ & true cost and bid of user $i$\\
  $\theta_i,\hat{\theta_i}$ & true type and strategy of user $i$\\
  $\mathcal{S},\mathcal{S}'$ & set of selected users and sample set\\
  $p_i,u_i$ & payment and utility of user $i$\\
  $V(\mathcal{S})$ & value function of the crowdsourcer over $\mathcal{S}$\\
  $V_i(\mathcal{S})$ & marginal value of user $i$ over $\mathcal{S}$\\
  $\rho^*$ & density threshold\\
  $\delta$ & parameter used for computing the density threshold\\
  $\omega$ & parameter assumed on users' value\\
  \hline
\end{tabularx}
\end{center}
\end{table}

\section{Online Mechanism under Zero Arrival-departure Interval Case}
\label{sec:special case}
Firstly, we relax one assumption of the problem: considering a special case where the arrival time of each user equals to its departure time.
In this case, each user is impatient since the decision must be made immediately once it arrives.
Note that achieving time-truthfulness is trivial in this case.
It is because that any user has no incentive to report a later arrival time or an earlier departure time than its true arrival/departure time, since the user cannot perform any sensing task or obtain a payment after it departs.
In this section, we present an online mechanism satisfying all desirable properties under this special case (called \emph{zero arrival-departure interval} case later), without considering the time-truthfulness.
Then, in Section \ref{sec:general case} we revise this mechanism and prove the revised one satisfies all desirable properties including the time-truthfulness under the general case without \emph{zero arrival-departure interval} assumption.
To facilitate understanding, in this section it is also assumed that no two users have the same arrival time.
Note that this assumption can also be easily removed according to the revised mechanism in Section \ref{sec:general case}.

\subsection{Mechanism Design}
An online mechanism needs to overcome several nontrivial challenges: firstly, the users' costs are unknown and need to be reported in a truthful manner; secondly, the total payment cannot exceed the crowdsourcer's budget; finally, and most important, the mechanism needs to cope with the online arrival of the users.
Previous solutions of online auctions and generalized secretary problems \cite{hajiaghayi2004adaptive,babaioff2008online,kleinberg2005multiple,babaioff2007knapsack,bateni2010submodular} always achieve desirable outcomes in online settings via a two-stage sampling-accepting process or containing such process: the first batch of applicants is rejected and used as the sample which enables making an informed decision on whether accepting the rest of applicants.
However, these solutions cannot guarantee the consumer sovereignty, since the first batch of applicants has no chance to win the auction no matter how low its cost is.
It can lead to undesirable effects in our problem: automatically rejecting the first batch of users encourages users to arrive late; in other words, those users arriving early have no incentive to report their bids to the crowdsourcer, which may delay the users' completion or even result in task starvation.

To address the above challenges, we design our online mechanism, \emph{OMZ}, based on a \emph{multiple-stage sampling-accepting} process.
The mechanism dynamically increases the sample size and learns a \emph{density threshold} used for future decision, while increasing the \emph{stage-budget} it uses for allocation at various stages.
The whole process is illustrated in Algorithm \ref{alg:OMZ}.
Firstly, we divide all of $T$ time steps into $(\lfloor \log_2 T \rfloor+1)$ stages: $\{1, 2, \ldots, \lfloor \log_2 T \rfloor, \lfloor \log_2 T \rfloor+1\}$.
The stage $i$ ends at time step $T'=\lfloor 2^{i-1}T / 2^{\lfloor \log_2 T \rfloor} \rfloor$.
Correspondingly, the stage-budget for the $i$-th stage is allocated as $B'=2^{i-1}B / 2^{\lfloor \log_2 T \rfloor}$.
Fig. \ref{fig-stages} is an illustration when $T=8$.
When a stage is over, we add all users who have arrived into the sample set $\mathcal{S}'$, and compute a density threshold $\rho^*$ according to the information of samples and the allocated stage-budget $B'$.
This density threshold is computed by calling the \textbf{GetDensityThreshold} algorithm (to be elaborated later), and used for making decision at the next stage.
Specially, when the last stage $i = \lfloor \log_2 T \rfloor+1$ comes, the density threshold has been computed according to the information of all users arriving before time step $\lfloor T/2 \rfloor$, and the allocated stage-budget $B/2$.
\begin{figure}[!t]
\centering{
\includegraphics[width=3.5in]{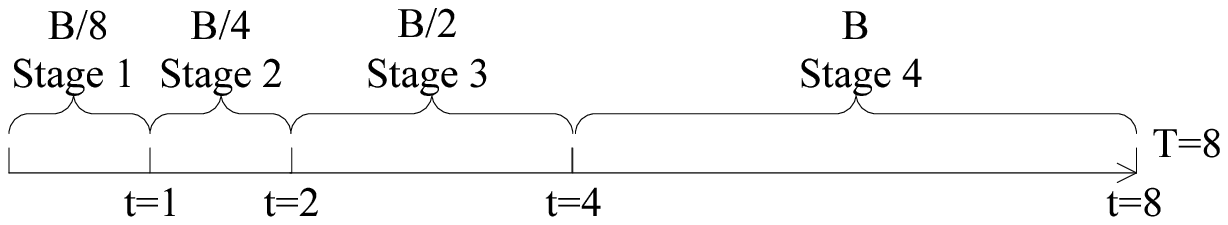}}
\caption{Illustration of a multiple-stage sampling-accepting process when $T=8$.}
\label{fig-stages}
\vspace{-10pt}
\end{figure}
\begin{algorithm}
\caption{Online Mechanism under Zero Arrival-departure Interval Case (\emph{OMZ})}
\label{alg:OMZ}
\KwIn{Budget constraint $B$, deadline $T$}
$(t,T',B',\mathcal{S}',\rho^*,\mathcal{S}) \leftarrow ( 1, \frac{T}{2^{\lfloor \log_2 T \rfloor}}, \frac{B}{2^{\lfloor \log_2 T \rfloor}}, \emptyset, \epsilon, \emptyset)$\;
\While{$t\leq T$}
{
    \If{there is a user $i$ arriving at time step $t$}
    {
        \uIf{$b_i\leq V_i(\mathcal{S})/\rho^* \leq B'-\sum_{j \in \mathcal{S}} p_j$}
        {
            $p_i\leftarrow V_i(\mathcal{S})/\rho^*$; $\mathcal{S}\leftarrow \mathcal{S} \cup \{i\}$\;
        }
        \lElse
        {
            $p_i\leftarrow 0$\;
        }
        $\mathcal{S}'\leftarrow \mathcal{S}' \cup \{i\}$\;
    }
    \If{$t=\lfloor T' \rfloor$}
    {
        $\rho^*\leftarrow \textbf{GetDensityThreshold}(B',\mathcal{S}')$\;
        $T'\leftarrow 2T'$; $B'\leftarrow 2B'$\;
    }
    $t\leftarrow t+1$\;
}
\end{algorithm}

Given a set of selected users $\mathcal{S}$, the \emph{marginal value} of user $i \notin \mathcal{S}$ is $V_i(\mathcal{S})=V(\mathcal{S} \cup \{i\})-V(\mathcal{S})$, and its \emph{marginal density} is $V_i(\mathcal{S})/b_i$.
When a new user $i$ arrives, the mechanism allocates tasks to it as long as its marginal density is not less than the current threshold density $\rho^*$, and the allocated stage-budget $B'$ has not been exhausted.
Meanwhile, we give user $i$ a payment
\[p_i=V_i(\mathcal{S})/\rho^*,\]
and add this user to the set of selected users $\mathcal{S}$.
To start the mechanism, we initially set a small density threshold $\epsilon$, which is used for making decision at the first stage.

Since each stage maintains a common density threshold, it is natural to adopt a \emph{proportional share} allocation rule to compute the density threshold from the sample set $\mathcal{S}'$ and the allocated stage-budget $B'$.
As illustrated in Algorithm \ref{alg:get threshold}, the computation process adopts a greedy strategy.
Users are sorted according to their increasing marginal densities.
In this sorting the $(i + 1)$-th user is the user $j$ such that $V_j(\mathcal{S}_i)/b_j$ is maximized over $\mathcal{S}'\backslash \mathcal{S}_i$, where $\mathcal{S}_i=\{1,2,\ldots,i\}$ and $\mathcal{S}_0=\emptyset$.
Considering the submodularity of $V$, this sorting implies that:
\[\frac{V_1(\mathcal{S}_0)}{b_1} \geq \frac{V_2(\mathcal{S}_1)}{b_2} \geq \dots \geq \frac{V_{|\mathcal{S}'|}(\mathcal{S}_{|\mathcal{S}'|-1})}{b_{|\mathcal{S}'|}}.\]
Find the largest $k$ such that $b_k \leq \frac{V_k(\mathcal{S}_{k-1}) B}{V(\mathcal{S}_k)}$.
The set of selected users is $\mathcal{S}_k=\{1,2,\ldots,k\}$.
Finally, we set the density threshold to be $\frac{V(\mathcal{S}_k)}{\delta B'}$.
Here we set $\delta>1$ to obtain a slight underestimate of the density threshold for guaranteeing enough users selected and avoiding the waste of budget.
Later we will fix the value of $\delta$ elaborately to enable the mechanism achieving a constant competitive ratio.
\begin{algorithm}
\caption{GetDensityThreshold}
\label{alg:get threshold}
\KwIn{Stage-budget $B'$, sample set $\mathcal{S}'$}
$\mathcal{J}\leftarrow \emptyset$; $i \leftarrow \arg\max_{j\in \mathcal{\mathcal{S}'}}(V_j(\mathcal{J})/b_j)$\;
\While{$b_i\leq \frac{V_i(\mathcal{J}) B'}{V(\mathcal{J}\cup \{i\})}$}
{
    $\mathcal{J}\leftarrow \mathcal{J}\cup \{i\}$\;
    $i\leftarrow \arg\max_{j\in \mathcal{\mathcal{S}'}\backslash \mathcal{J}}(V_j(\mathcal{J})/b_j)$\;
}
$\rho\leftarrow V(\mathcal{J})/B'$\;
\Return $\rho/\delta$;
\end{algorithm}

In the following, we use an example to illustrate how the \emph{OMZ} mechanism works.
\begin{example}
\label{exmaple1}
Consider a crowdsourcer with the budget constraint $B=16$ and the deadline $T=8$.
There are five users arriving online before the deadline with types $\theta_i=(a_i,d_i,\Gamma_i,c_i)$, where $a_i=d_i$, and $\Gamma_i$ can be omitted by assuming that each user has the same marginal value 1.
Here the types $(a_i,d_i,c_i)$ of the five users are: $\theta_1=(1,1,2)$, $\theta_2=(2,2,4)$, $\theta_3=(4,4,5)$, $\theta_4=(6,6,1)$, and $\theta_5=(7,7,3)$.
\end{example}

We set $\epsilon=1/2$ and $\delta=1$. Then the \emph{OMZ} mechanism works as follows.
\begin{itemize*}
    \item[$\diamond$] $t=1$: $(T',B',\mathcal{S}',\rho^*,\mathcal{S})=(1,2,\emptyset,1/2,\emptyset)$, $V_1(\mathcal{S})/b_1=1/2$, thus $p_1=2$, $\mathcal{S}=\{1\}$, $\mathcal{S}'=\{1\}$. Update the density threshold: $\rho^*=1/2$.
    \item [$\diamond$] $t=2$: $(T',B',\mathcal{S}',\rho^*,\mathcal{S})=(2,4,\{1\},1/2,\{1\})$, $V_2(\mathcal{S})/b_2=1/4$, thus $p_2=0$, $\mathcal{S}'=\{1,2\}$. Update the threshold density: $\rho^*=1/4$.
    \item [$\diamond$] $t=4$: $(T',B',\mathcal{S}',\rho^*,\mathcal{S})=(4,8,\{1,2\},1/4,\{1\})$, $V_3(\mathcal{S})/b_3=1/5$, thus $p_3=0$, $\mathcal{S}'=\{1,2,3\}$. Update the density threshold: $\rho^*=1/4$.
    \item [$\diamond$] $t=6$: $(T',B',\mathcal{S}',\rho^*,\mathcal{S})=(8,16,\{1,2,3\},1/4,\{1\})$, $V_4(\mathcal{S})/b_4=1$, thus $p_4=4$, $\mathcal{S}=\{1,4\}$, $\mathcal{S}'=\{1,2,3,4\}$.
    \item [$\diamond$] $t\!=\!7\!$: $\!(T',B',\mathcal{S}',\rho^*,\mathcal{S})\!=\!(8,16,\{1,2,3,4\},1/4,\{1,4\})$, $V_5(\mathcal{S})/b_5=1/3$, thus $p_5=4$. Finally, the set of selected users is $\mathcal{S}=\{1,4,5\}$, and the payments of these selected 3 users are 2, 4, 4 respectively.
\end{itemize*}

\subsection{Mechanism Analysis}
\label{subsec:mechanism anlysis}
In the following, we will firstly prove that the \emph{OMZ} mechanism satisfies the computational efficiency (Lemma \ref{lemma:computational efficiency}), individual rationality (Lemma \ref{lemma:individual rationality}), budget feasibility (Lemma \ref{lemma:budget feasibility}), cost-truthfulness (Lemma \ref{lemma:cost-truthfulness}), and the consumer sovereignty (Lemma \ref{lemma:consumer sovereignty}).
Then, we will prove that the \emph{OMZ} mechanism can achieve a constant competitive ratio under both the i.i.d. model (Lemma \ref{lemma:average case}) and the secretary model (Lemma \ref{lemma:worst case}) by elaborately fixing different values of $\delta$.

\begin{lemma}
\label{lemma:computational efficiency}
The OMZ mechanism is computationally efficient.
\end{lemma}
\begin{proof}
Since the mechanism runs online, we only need to focus on the computation complexity at each time step $t \in \{1, \ldots, T\}$.
Computing the marginal value of user $i$ takes $O(|\Gamma_i|)$ time, which is at most $O(m)$.
Thus, the running time of computing the allocation and payment of user $i$ (lines 3-8) is bounded by $O(m)$.
Next, we analyze the complexity of computing the density threshold, namely Algorithm \ref{alg:get threshold}.
Finding the user with maximum marginal density takes $O(m|\mathcal{S}'|)$ time.
Since there are $m$ tasks and each selected user should contribute at least one new task, the number of winners is at most $\min \{m,|\mathcal{S}'|\}$.
Thus, the running time of Algorithm \ref{alg:get threshold} is bounded by $O(m|\mathcal{S}'|\min \{m,|\mathcal{S}'|\})$.
Therefore, the computation complexity at each time step (lines 3-13) is bounded by $O(m|\mathcal{S}'| \min \{m,|\mathcal{S}'|\})$.
At the last stage, the sample set $\mathcal{S}'$ has the maximum number of samples, being $n/2$ with high probability.
Thus, the computation complexity at each time step is bounded by $O(mn \min \{m,n\})$.
\end{proof}

Note that the above analysis of running time is very conservative.
In practice, the running time of computing the marginal value, $O(|\Gamma_i|)$, is much less than $O(m)$.
In addition, the running time of the \emph{OMZ} mechanism will increase linearly with $n$ especially when $n$ is large.

\begin{lemma}
\label{lemma:individual rationality}
The OMZ mechanism is individually rational.
\end{lemma}
\begin{proof}
From the lines 4-6 of Algorithm \ref{alg:OMZ}, we can see that $p_i\geq b_i$ if $i\in \mathcal{S}$, otherwise $p_i=0$. Therefore, we have $u_i\geq 0$.
\end{proof}

\begin{lemma}
\label{lemma:budget feasibility}
The OMZ mechanism is budget feasible.
\end{lemma}
\begin{proof}
At each stage $i\in \{1, 2, \ldots, \lfloor \log_2 T \rfloor, \lfloor \log_2 T \rfloor+1\}$, the mechanism uses a stage-budget of $B'=\frac{2^{i-1}B}{2^{\lfloor \log_2 T \rfloor}}$.
From the lines 4-5 of Algorithm \ref{alg:OMZ}, we can see that it is guaranteed that the current total payment does not exceed the stage-budget $B'$.
Specially, the budget constraint of the last stage is $B$.
Therefore, every stage is budget feasible, and when the deadline $T$ arrives, the total payment does not exceed $B$.
\end{proof}

Designing a cost-truthful mechanism relies on the rationale of \emph{bid-independence}.
Let $b_{-i}$ denote the sequence of bids arriving before the $i$-th bid $b_i$, i.e., $b_{-i}=(b_1,\ldots,b_{i-1})$.
We call such a sequence \emph{prefixal}.
Let $p'$ be a function from prefixal sequences to prices (non-negative real numbers).
We extend the definition of bid-independence \cite{goldberg2006competitive} to the online scenario as follows.

\begin{definition}[Bid-independent Online Auction]
An online auction is called bid-independent if the allocation and payment rules for each player $i$ satisfy:
\begin{enumerate}[a)]\setlength{\itemindent}{1.1em}
    \item The auction constructs a price schedule $p'(b_{-i})$;
    \item If $p'(b_{-i})\geq b_i$, player $i$ wins at price $p_i=p'(b_{-i})$;
    \item Otherwise, player $i$ is rejected, and $p_i=0$.
\end{enumerate}
\end{definition}
\begin{proposition}
(\cite{bar2002incentive}, Proposition 2.1) An online auction is cost-truthful if and only if it is bid-independent.
\end{proposition}
\begin{lemma}
\label{lemma:cost-truthfulness}
The OMZ mechanism is cost-truthful.
\end{lemma}
\begin{proof}
To see that bid-independent auctions are cost-truthful, here we consider a user $i$ that arrives at some stage for which the density threshold was set to $\rho^*$.
If by the time the user arrives there are no remaining budget, then the user's cost declaration will not affect the allocation of the mechanism and thus cannot improve its utility by submitting a false cost.
Otherwise, assume there are remaining budget by the time the user arrives.
In case $c_i \leq V_i(\mathcal{S})/\rho^*$, reporting any cost below $V_i(\mathcal{S})/\rho^*$ wouldn't make a difference in the user's allocation and payment and its utility would be $V_i(\mathcal{S})/\rho^*-c_i\geq 0$.
Declaring a cost above $V_i(\mathcal{S})/\rho^*$ would make the worker lose the auction, and its utility would be 0.
In case $c_i > V_i(\mathcal{S})/\rho^*$, declaring any cost above $V_i(\mathcal{S})/\rho^*$ would leave the user unallocated with utility 0.
If the user declares a cost lower than $V_i(\mathcal{S})/\rho^*$ it will be allocated.
In such a case, however, its utility will be negative.
Hence the user's utility is always maximized by reporting its true cost: $b_i=c_i$.
\end{proof}

\begin{lemma}
\label{lemma:consumer sovereignty}
The OMZ mechanism satisfies the consumer sovereignty.
\end{lemma}
\begin{proof}
Each stage is an accepting process as well as a sampling process ready for the next stage.
As a result, users are not automatically rejected during the sampling process, and are allocated as long as their marginal densities are not less than the current density threshold, and the allocated stage-budget has not been exhausted.
\end{proof}

Before analyzing the competitiveness of the \emph{OMZ} mechanism, we firstly introduce an offline mechanism proposed by Singer \cite{singer2010budget}, which is proved to satisfy \emph{computational efficiency}, \emph{individual rationality}, \emph{budget feasibility}, and \emph{truthfulness}.
This mechanism does not have knowledge about users' costs, but it is an offline mechanism, i.e., all users submit their bids to the mechanism and wait for the mechanism to collect all the bids and decide on an allocation.
This mechanism has been proved to be $O$(1)-\emph{competitive} in maximizing the value of services received under budget constraint compared with the optimal solution.
Therefore, we only need to prove that the \emph{OMZ} mechanism has a constant competitive ratio compared with this offline mechanism, then the \emph{OMZ} mechanism will also be $O$(1)-\emph{competitive} compared with the optimal solution.
Note that in the offline scenario satisfying the \emph{time-truthfulness} and the \emph{consumer sovereignty} is trivial, since all decisions are made after all users' information is submitted to the crowdsourcer.

\begin{algorithm}
\caption{Proportional Share Mechanism (Offline) \protect\cite{singer2010budget}}
\label{alg:offline}
\KwIn{Budget constraint $B$, User set $\mathcal{U}$}
\tcc{Winner selection phase}
$\mathcal{S}\leftarrow \emptyset$; $i \leftarrow \arg\max_{j\in \mathcal{U}}(V_j(\mathcal{S})/b_j)$\;
\While{$b_i\leq \frac{V_i(\mathcal{S}) B}{V(\mathcal{S}\cup \{i\})}$}
{
    $\mathcal{S}\leftarrow \mathcal{S}\cup \{i\}$\;
    $i\leftarrow \arg\max_{j\in \mathcal{U}\backslash \mathcal{S}}(V_j(\mathcal{S})/b_j)$\;
}
\tcc{Payment determination phase}
\lForEach{$i\in \mathcal{U}$}{$p_i\leftarrow 0$}\;
\ForEach{$i\in \mathcal{S}$}
{
    $\mathcal{U'}\leftarrow \mathcal{U}\backslash \{i\}$; $\mathcal{Q}\leftarrow \emptyset$\;
    \Repeat{$b_{i_j}\leq \frac{V_{i_j}(\mathcal{Q}_{j-1}) B}{V(\mathcal{Q})}$}
    {
        $i_j\leftarrow \arg\max_{j\in \mathcal{U'}\backslash \mathcal{Q}}(V_j(\mathcal{Q})/b_j)$\;
        $p_i\leftarrow \max \{p_i,\min\{b_{i(j)},\eta_{i(j)}\}\}$\;
    }
}
\Return $(\mathcal{S},p)$;\
\end{algorithm}
The offline mechanism adopts a \emph{proportional share allocation rule}. As described in Algorithm \ref{alg:offline}, it consists of two phases: the \emph{winner selection} phase and the \emph{payment determination} phase. The \emph{winner selection} phase has the same working process as Algorithm \ref{alg:get threshold}. In the payment determination phase, we compute the payment $p_i$ for each winner $i\in \mathcal{S}$.
To compute the payment for user $i$, we sort the users in $\mathcal{U}\backslash \{i\}$ similarly:
\[\frac{V_{i_1}(\mathcal{Q}_0)}{b_{i_1}} \geq \frac{V_{i_2}(\mathcal{Q}_1)}{b_{i_2}} \geq \dots \geq \frac{V_{i_{n-1}}(\mathcal{Q}_{n-2})}{b_{i_{n-1}}},\]
where $V_{i_j}(\mathcal{Q}_{j-1})=V(\mathcal{Q}_{j-1} \cup \{i_j\})-V(\mathcal{Q}_{j-1})$ denotes the marginal value of the $j$-th user and $\mathcal{Q}_j$ denotes the first $j$ users according to this sorting over $\mathcal{U}\backslash \{i\}$ and $\mathcal{Q}_0=\emptyset$.
The marginal value of user $i$ at position $j$ is $V_{i(j)}(\mathcal{Q}_{j-1})=V(\mathcal{Q}_{j-1} \cup \{i\})-V(\mathcal{Q}_{j-1})$.
Let $k'$ denote the position of the last user $i_j \in \mathcal{U}\backslash \{i\}$, such that $b_{i_j} \leq V_{i_j}(\mathcal{Q}_{j-1})B/V(\mathcal{Q}_j)$.
For brevity we will write $b_{i(j)}=V_{i(j)}(\mathcal{Q}_{j-1}) b_{i_j}/V_{i_j}(\mathcal{Q}_{j-1})$, and $\eta_{i(j)}=V_{i(j)}(\mathcal{Q}_{j-1}) B/V(\mathcal{Q}_{j-1} \cup \{i\})$.
In order to guarantee the truthfulness, each winner should be paid the critical value, which means that user $i$ would not
win the auction if it bids higher than this value.
Thus, the payment for user $i$ should be the maximum of these $k' + 1$ prices:
\[p_i=\max_{j\in[k'+1]}\{\min \{b_{i(j)},\eta_{i(j)}\}\}.\]

Let $Z$ be the set of selected users $\mathcal{S}$ computed by Algorithm \ref{alg:offline}, and the value of $Z$ is $V(Z)$.
The value density of $Z$ is $\rho=V(Z)/B$.
Define $Z_1$ and $Z_2$ as the subsets of $Z$ that appears in the first and second half of the input stream, respectively.
When the stage $\lfloor \log_2 T \rfloor$ is over, we obtain the sample set $\mathcal{S}'$ consisting of all users arriving before the time $\lfloor T/2 \rfloor$.
Thus, we have $Z_1=Z\cap \mathcal{S}'$, and $Z_2=Z\cap \{\mathcal{U}\backslash \mathcal{S}'\}$.
Let $Z_1'$ denote the set of selected users computed by Algorithm \ref{alg:get threshold} based on the sample set $\mathcal{S}'$ and the allocated stage-budget $B/2$, and the value of $Z_1'$ is $V(Z_1')$.
The density of $Z_1'$ is $\rho_1'=2V(Z_1')/B$.
The density threshold of the last stage is $\rho^*=\rho_1'/\delta$.
Let $Z_2'$ denote the set of selected users computed by Algorithm \ref{alg:OMZ} at the last stage.
Assume that the value of each user is at most $V(Z)/\omega$, where the parameter $\omega$ will be fixed later.

\subsubsection{\textbf{Competitiveness Analysis under the I.I.D. Model}}
Since the costs and values of all users in $\mathcal{U}$ are i.i.d., they can be selected in the set $Z$ with the same probability.
Note that the sample set $\mathcal{S}'$ is a random subset of $\mathcal{U}$ since all users arrive in a random order.
Therefore the number of users from $Z$ in the sample set $\mathcal{S}'$ follows a hypergeometric distribution $H(n/2,|Z|,n)$.
Thus, we have $\mathbb{E}[|Z_1|]=\mathbb{E}[|Z_2|]=|Z|/2$.
The value of each user can be seen as an independent identically distributed random variable, and because of the submodularity of $V(\mathcal{S})$, it can be derived that: $\mathbb{E}[V(Z_1)]=\mathbb{E}[V(Z_2)]\geq V(Z)/2$.
The expected total payments to the users from both $Z_1$ and $Z_2$ are $B/2$.
Since $V(Z_1')$ is computed with the stage-budget $B/2$, it can be derived that: $\mathbb{E}[V(Z_1')]\geq \mathbb{E}[V(Z_1)]\geq V(Z)/2$, and $\mathbb{E}[\rho_1'] \geq \rho$, where the first inequality follows from the fact that $V(Z_1')$ is the optimal solution computed by Algorithm \ref{alg:get threshold} with stage-budget $B/2$ according to the \emph{proportional share allocation} rule.
Therefore, we only need to prove that the ratio of $\mathbb{E}[V(Z_2')]$ to $\mathbb{E}[V(Z_1')]$ is at least a constant, then the \emph{OMZ} mechanism will also have a constant expected competitive ratio compared with the offline mechanism.
\begin{lemma}
\label{lemma:average case}
For sufficiently large $\omega$, the ratio of $\mathbb{E}[V(Z_2')]$ to $\mathbb{E}[V(Z_1')]$ is at least a constant. Specially, this ratio approaches $1/4$ as $\omega \rightarrow \infty$ and $\delta \rightarrow 4$.
\end{lemma}

The proof of Lemma \ref{lemma:average case} is given in Appendix B.

\subsubsection{\textbf{Competitiveness Analysis under the Secretary Model}}
\begin{lemma}
\label{lemma:Z1Z2}
(\cite{bateni2010submodular}, Lemma 16) For sufficiently large $\omega$, the random variable $|V(Z_1)-V(Z_2)|$ is bounded by $V(Z)/2$ with a constant probability.
\end{lemma}

Note that a non-negative submodular function is also a subadditive function, so we have $V(Z_1)+V(Z_2)\geq V(Z)$. Thus, Lemma \ref{lemma:Z1Z2} can be easily extended to the following corollary.
\begin{corollary}
\label{corollary}
For sufficiently large $\omega$, both $V(Z_1)$ and $V(Z_2)$ are at least $V(Z)/4$ with a constant probability.
\end{corollary}

\begin{lemma}
\label{lemma:value}
Given a sample set $\mathcal{S}'$, the total value of selected users computed by Algorithm \ref{alg:get threshold} with the budget $B'/2$ is at least a half of that computed with the budget $B'$.
\end{lemma}
The proof of Lemma \ref{lemma:value} is given in Appendix C.

Note that the total value of selected users from the sample set $\mathcal{S}'$ computed by Algorithm 2 with the budget $B$ is not less than $V(Z_1)$.
Thus, considering Corollary \ref{corollary} and Lemma \ref{lemma:value}, it can be derived that: $V(Z_1')\geq V(Z_1)/2 \geq V(Z)/8$.
Therefore, it only needs to prove that the ratio of $V(Z_2')$ to $V(Z_1')$ is at least a constant, then the \emph{OMZ} mechanism will also have a constant competitive ratio compared with the offline mechanism.

\begin{lemma}
\label{lemma:worst case}
For sufficiently large $\omega$, the ratio of $V(Z_2')$ to $V(Z_1')$ is at least a constant. Specially, this ratio approaches $1/12$ as $\omega \rightarrow \infty$ and $\delta \rightarrow 12$.
\end{lemma}

The proof of Lemma \ref{lemma:worst case} is given in Appendix D.

From the above analysis, we know that the $OMZ$ mechanism has a competitive factor of at least 8 (96) of the offline \emph{proportional share} solution under the i.i.d. model (the secretary model).
While the competitive ratio may seem large, we emphasize that our goal is to show that the $OMZ$ mechanism is indeed $O$(1)-\emph{competitive}, and thus its performance guarantee is independent of the parameters of the problem (e.g. number of users, their costs, the tasks they can complete, etc.).
We will later show that the mechanism performs well in practice (see Section \ref{sec:performance evaluation}), implying that bounded competitive ratio serves as a good guide for designing such mechanisms.

\begin{theorem}
The OMZ mechanism satisfies computational efficiency, individual rationality, budget feasibility, truthfulness, consumer sovereignty, and constant competitiveness under the zero arrival-departure interval case.
\end{theorem}

\section{Online Mechanism under General Case}
\label{sec:general case}
In this section, we consider the general case where each user may have a \emph{non-zero arrival-departure interval}, and there may be multiple online users in the auction simultaneously.
Firstly, we change the settings of Example \ref{exmaple1} to show that the \emph{OMZ} mechanism is \emph{not} time-truthful under the general case.

\begin{example}
\label{example2}
All the settings are the same as Example \ref{exmaple1} except for that user 1 has a non-zero arrival-departure interval, $a_1<d_1$.
Specially, the type of user 1 is $\theta_1=(1,5,2)$.
\end{example}

In this example, if user 1 report its type truthfully, then it will obtain the payment 2 according to the \emph{OMZ} mechanism.
However, if user 1 delays announcing its arrival time and reports $\theta_1'=(5,5,2)$, then it will improve its payment to 8 according to the \emph{OMZ} mechanism (the detailed computing process is omitted).

In the following, we will present a new online mechanism, \emph{OMG}, and prove that it satisfies all six desirable properties under the general case.

\subsection{Mechanism Design}
In order to hold several desirable properties of the \emph{OMZ} mechanism, we adopt a similar algorithm framework under the general case.
At the same time, in order to guarantee the \emph{cost-} and \emph{time-truthfulness}, it is necessary to modify the \emph{OMZ} mechanism based on the following principles.
Firstly, any user is added to the sample set only when it departs; otherwise, the bid-independence will be destroyed if its arrival-departure time spans multiple stages, because a user can indirectly affect its payment now.
Secondly, if there are multiple users who have not yet departed at some time, we sort these online users according to their marginal values, instead of their marginal densities, and preferentially select those users with higher marginal value.
In this way, the bid-independence can be held.
Thirdly, whenever a new time step arrives, it scans through the list of users who have not yet departed and selects those whose marginal densities are not less than the current density threshold under the stage-budget constraint, even if some arrived much earlier.
At the departure time of any user who was selected as a winner, the user is paid for a price equal to the maximum price attained during the user's reported arrival-departure interval, even if this price is larger than the price at the time step when the user was selected as a winner.
\begin{algorithm}
\caption{Online Mechanism under General Case (\emph{OMG})}
\label{alg:OMG}
\KwIn{ Budget constraint $B$, deadline $T$}
$(t,T',B',\mathcal{S}',\rho^*,\mathcal{S}) \leftarrow ( 1, \frac{T}{2^{\lfloor \log_2 T \rfloor}}, \frac{B}{2^{\lfloor \log_2 T \rfloor}}, \emptyset, \epsilon, \emptyset)$\;
\While{$t\leq T$}
{
    Add all new users arriving at time step $t$ to a set of online users $\mathcal{O}$; $\mathcal{O}'\leftarrow \mathcal{O}\setminus S$\;
    \Repeat{$\mathcal{O}'=\emptyset$}
    {
        $i\leftarrow \arg\max_{j\in \mathcal{O}'}(V_j(\mathcal{S}))$\;

        \uIf{$b_i\leq V_i(\mathcal{S})/\rho^* \leq B'-\sum_{j \in \mathcal{S}} p_j$}
        {
            $p_i\leftarrow V_i(\mathcal{S})/\rho^*$; $\mathcal{S}\leftarrow \mathcal{S} \cup \{i\}$\;
        }
        \lElse
        {
            $p_i\leftarrow 0$\;
        }
        $\mathcal{O}'\leftarrow \mathcal{O}'\setminus \{i\}$\;
    }
    Remove all users departing at time step $t$ from $\mathcal{O}$, and add them to $\mathcal{S}'$\;
    \If{$t=\lfloor T' \rfloor$}
    {
        $\rho^*\leftarrow \textbf{GetDensityThreshold}(B',\mathcal{S}')$\;
        $T'\leftarrow 2T'$; $B'\leftarrow 2B'$; $\mathcal{O}'\leftarrow \mathcal{O}$\;
        \Repeat{$\mathcal{O}'=\emptyset$}
        {
            $i\leftarrow \arg\max_{j\in \mathcal{O}'}(V_j(\mathcal{S}\setminus \{j\}))$\;
            \If{$b_i\leq V_i(\mathcal{S}\setminus \{i\})/\rho^* \leq B'-\sum_{j \in \mathcal{S}} p_j +p_i$ {\bf and} $V_i(\mathcal{S}\setminus \{i\})/\rho^* > p_i$}
            {
                $p_i\leftarrow V_i(\mathcal{S}\setminus \{i\})/\rho^*$\;
                \lIf{$i \notin \mathcal{S}$}
                {
                    $\mathcal{S}\leftarrow \mathcal{S} \cup \{i\}$\;
                }
            }
            $\mathcal{O}'\leftarrow \mathcal{O}'\setminus \{i\}$\;
        }
    }
    $t\leftarrow t+1$\;
}
\end{algorithm}

According to the above principles, we design the \emph{OMG} mechanism satisfying all desirable properties under the general case, as described in Algorithm \ref{alg:OMG}.
Specially, we consider two cases as follows.

The first case is when the current time step $t$ is not at the end of any stage.
In this case, the density threshold remains unchanged.
The following operations (the lines 3-11 in Algorithm \ref{alg:OMG}) are performed.
Firstly, all new users arriving at time step $t$ are added to a set of online users $\mathcal{O}$.
Then we make decision on whether to select these online users one by one in the order of their marginal values; the users with higher marginal values will be selected first.
If an online user $i$ has been selected as a winner before time step $t$, we need not to make decision on it again because it is impossible to obtain a higher payment than before (to be proved later in Lemma \ref{lemma:truthfulness2}).
Otherwise, we need to make decision on it again: if its marginal density is not less than the current density threshold, and the allocated stage-budget has not been exhausted, it will be selected as a winner.
Meanwhile, we give user $i$ a payment $p_i=V_i(\mathcal{S})/\rho^*$, and add it to the set of selected users $\mathcal{S}$.
Finally, we remove all users departing at time step $t$ from $\mathcal{O}$, and add them to the sample set $\mathcal{S}'$.

The second case is when the current time step is just at the end of some stage.
In this case, the density threshold will be updated.
The mechanism works as the lines 13-22.
We need to make decision on whether to select these online users, and at what prices, one by one in the order of their marginal values, no matter whether they have ever been selected as the winners before time step $t$.
As shown in the lines 17-20, if user $i$ can obtain a higher payment than before, its payment will be updated.
Meanwhile, if user $i$ has never been selected as a winner before time step $t$, it will be added to the set $\mathcal{S}$.

Return to Example \ref{example2}.
If all of the five users report their types truthfully, then the \emph{OMG} mechanism works as follows.
\begin{itemize*}
    \item[$\diamond$] $t=1$: $(T',B',\mathcal{S}',\rho^*,\mathcal{S})=(1,2,\emptyset,1/2,\emptyset)$, $V_1(\mathcal{S})/b_1=1/2$, thus $p_1=2$, $\mathcal{S}=\{1\}$. Update the density threshold: $\rho^*=1/2$, $p_1$ remains unchanged.
    \item [$\diamond$] $t=2$: $(T',B',\mathcal{S}',\rho^*,\mathcal{S})=(2,4,\emptyset,1/2,\{1\})$, $V_2(\mathcal{S})/b_2=1/4$, thus $p_2=0$, $\mathcal{S}'=\{2\}$. Update the threshold density: $\rho^*=1/4$, increase $p_1$ to 4.
    \item [$\diamond$] $t=4$: $(T',B',\mathcal{S}',\rho^*,\mathcal{S})=(4,8,\{2\},1/4,\{1\})$, $V_3(\mathcal{S})/b_3=1/5$, thus $p_3=0$, $\mathcal{S}'=\{2,3\}$. Update the threshold density: $\rho^*=1/8$, increase $p_1$ to 8.
    \item [$\diamond$] $t=5$: user 1 departs, so $\mathcal{S}'=\{1,2,3\}$.
    \item [$\diamond$] $t=6$: $(T',B',\mathcal{S}',\rho^*,\mathcal{S})=(8,16,\{1,2,3\},1/8,\{1\})$, $V_4(\mathcal{S})/b_4=1$, thus $p_4=8$, $\mathcal{S}=\{1,4\}$, $\mathcal{S}'=\{1,2,3,4\}$.
    \item [$\diamond$] $t\!=\!7\!$: $\!(T',B',\mathcal{S}',\rho^*,\mathcal{S})\!=\!(8,16,\{1,2,3,4\},1/8,\{1,4\})$, $V_5(\mathcal{S})/b_5=1/3$, thus $p_5=0$, $\mathcal{S}'=\{1,2,3,4,5\}$.
\end{itemize*}

Thus, user 1 can obtain the payment 8 according to the \emph{OMG} mechanism.
Even if user 1 delays announcing its arrival time and reports $\theta_1'=(5,5,2)$, it still cannot improve its payment (the detailed computing process is omitted).
Therefore, the time-truthfulness can be guaranteed in this case.

\subsection{Mechanism Analysis}
It is easy to know that the \emph{OMG} mechanism holds the \emph{individual rationality} and the \emph{consumer sovereignty} as \emph{OMZ} (with almost the same proof).
Although it is hard to give a strict competitive ratio, it is easy to know that the \emph{OMG} mechanism still satisfies the \emph{constant competitiveness}, and only have slight value loss compared with \emph{OMZ}.
In the following, we prove that the \emph{OMG} mechanism also satisfies the \emph{computational efficiency}, the \emph{budget feasibility}, and most importantly, the \emph{cost-} and \emph{time-truthfulness}.
\begin{lemma}
\label{lemma:computational efficiency2}
The OMG mechanism is computationally efficient.
\end{lemma}
\begin{proof}
Different from \emph{OMZ}, the \emph{OMG} mechanism needs to compute the allocations and payments of multiple online users at each time step.
Thus, the running time of computing the allocations and payments at each time step is bounded by $O(m|\mathcal{O}|)<O(mn)$, where $|\mathcal{O}|$ is the number of online users.
The complexity of computing the density threshold is the same as that of \emph{OMZ}.
Thus, the computation complexity at each time step is the same as that of \emph{OMZ}, i.e., bounded by $O(mn\min\{m,n\})$.
\end{proof}

\begin{lemma}
\label{lemma:budget feasibility2}
The OMG mechanism is budget feasible.
\end{lemma}
\begin{proof}
From the lines 6-7 and 17-18 of Algorithm \ref{alg:OMG}, we can see that it is guaranteed that the current total payment does not exceed the stage-budget $B'$.
Note that in the line 17, $p_i$ is the price paid for user $i$ in the previous stage instead of the current stage, so it cannot lead to the overrun of the current stage-budget.
Therefore, every stage is budget feasible, and when the deadline $T$ arrives, the total payment does not exceed $B$.
\end{proof}

\begin{lemma}
\label{lemma:truthfulness2}
The OMG mechanism is cost- and time-truthful.
\end{lemma}

The proof of Lemma \ref{lemma:truthfulness2} is given in Appendix E.

\begin{theorem}
The OMG mechanism satisfies computational efficiency, individual rationality, budget feasibility, truthfulness, consumer sovereignty, and constant competitiveness under the general case.
\end{theorem}

\section{Performance Evaluation}
\label{sec:performance evaluation}
To evaluate the performance of our online mechanisms, we implemented the \emph{OMZ} and \emph{OMG} mechanisms, and compared them against the following three benchmarks.
The first benchmark is the (approximate) \emph{optimal} offline solution which has full knowledge about all users' types.
The problem in this scenario is essentially a \emph{budgeted maximum coverage problem}, which is a well-known NP-hard problem.
It is known that a greedy algorithm provides a $(1-1/e)$-approximation solution \cite{khullera1999budgeted}.
The second benchmark is the \emph{proportional share} mechanism in the offline scenario (Algorithm \ref{alg:offline}).
The third benchmark is the \emph{random} mechanism, which adopts a naive strategy, i.e., rewards users based on an uninformed fixed threshold density.
The performance metrics include the \emph{running time}, the \emph{crowdsourcer's value}, and the \emph{user's utility}.
\subsection{Simulation Setup}
The application scenario introduced in Section \ref{sec:introduction} is considered in our simulation.
Specially, we set the same simulation scenario as the reference \cite{sheng2012energy}, where a WiFi signal sensing application is considered.
As shown in Fig. \ref{fig-RoI} obtained from the Google Map, the RoI is located at Manhattan, NY, which spans 4 blocks from west to
east with a total length of 1.135km and 4 blocks from south to north with a total width of 0.319 km, and includes the 6th,7th,8th Avenues and the 45th, 46th, 47th Streets.
We divide each road in the RoI into multiple discrete PoIs with a uniform spacing of 1 meter, so the RoI consists of 4353 PoIs ($m=4352$) in total.
Without loss of generality, let the coverage requirement of each PoI be 1.
We set the deadline ($T$) to 1800s, and vary the budget ($B$) from 100 to 10000 with the increment of 100.
Users arrive according to a Poisson process in time with arrival rate $\lambda$.
We vary $\lambda$ from 0.2 to 1 with the increment of 0.2.
Whenever a user arrives, it is placed at a random location on the roads of the RoI.
The \emph{OMZ} mechanism is implemented under the zero arrival-departure interval case, and the \emph{OMG} mechanism is implemented under the general case where the arrival-departure interval of each user is uniformly distributed over $[0,300]$ seconds.
The sensing range ($R$) of each sensor is set to 7 meters.
The cost of each user is uniformly distributed over $[1,10]$.
The initial density threshold ($\epsilon$) of Algorithm \ref{alg:OMZ} and \ref{alg:OMG} is set to 1.
As we proved in Lemma \ref{lemma:average case}, when $\delta=4$ the OMZ mechanism is $O$(1)-competitive for sufficiently large $\omega$.
Meanwhile we note that $\omega$ increases with the number of users who have arrived.
Thus, for Algorithm \ref{alg:OMZ} and \ref{alg:OMG}, we set $\delta=1$ initially, and change it to $\delta=4$ once the size of the sample set exceeds a specified threshold.
Note that this threshold could be an empirical value for real applications.
In our simulation, we set this threshold to 240, because we observe that each user's value is at most $1/100$ of the total value when the number of users is larger than 240.
For the \emph{random} mechanism, we obtain the average performance of 50 such solutions for evaluations, where in each solution the threshold density was chosen at random from the range of 1 to 29 \footnotemark[1].

All the simulations were run on a PC with 1.7 GHz CPU and 8 GB memory.
Each measurement is averaged over 100 instances.
\footnotetext[1]{Each user can cover at most 29 PoIs, and its bid is at least 1, so its marginal density is at most 29.}
\begin{figure}[!t]
\centering{
\includegraphics[width=3.5in]{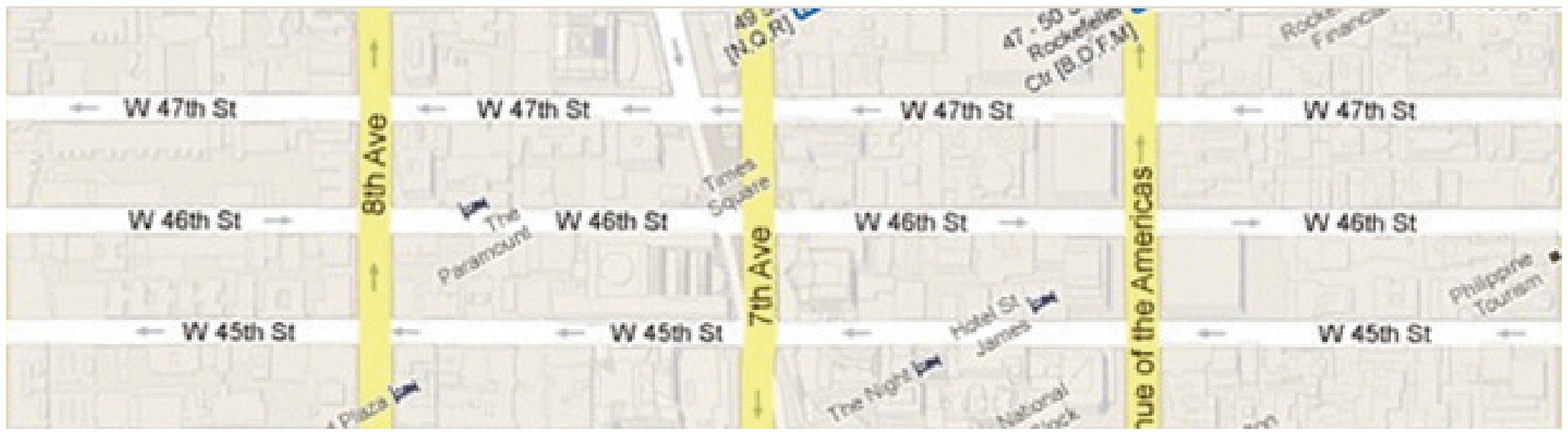}}
\caption{The region of interest.}
\label{fig-RoI}
\vspace{-10pt}
\end{figure}
\subsection{Evaluation Results}
\underline{Running Time}:
Fig. \ref{fig_runningTime} shows the running time of the \emph{OMZ} and \emph{OMG} mechanisms.
Specially, Fig. \ref{fig_runningTime1} plots the running time at different stages while $\lambda=0.6$\footnotemark[2].
\footnotetext[2]{As we proved in Lemma \ref{lemma:computational efficiency}, the computation complexity is dominated by computing
the density threshold, so only the running time at the end time of each stage is plotted.}
Fig. \ref{fig_runningTime2} plots the running time at the last stage with different arrival rates.
Both the \emph{OMZ} and \emph{OMG} mechanisms have similar performance while \emph{OMG} outperforms \emph{OMZ} slightly.
Note that the size of the sample set ($S'$) increases linearly with the time $t$ and the arrival rate $\lambda$, so Fig. \ref{fig_runningTime} implies the relationship between the running time and the number of users arriving before $T$.
Thus, from Fig. \ref{fig_runningTime} we can infer that the running time increases linearly with the number of users ($n$), which is consistent with our analysis in Section \ref{subsec:mechanism anlysis}.
\begin{figure}[!t]
  \centering{
  \subfigure[At different stages ($\lambda=0.6$) ]{
    \label{fig_runningTime1}
    \includegraphics[width=1.70in]{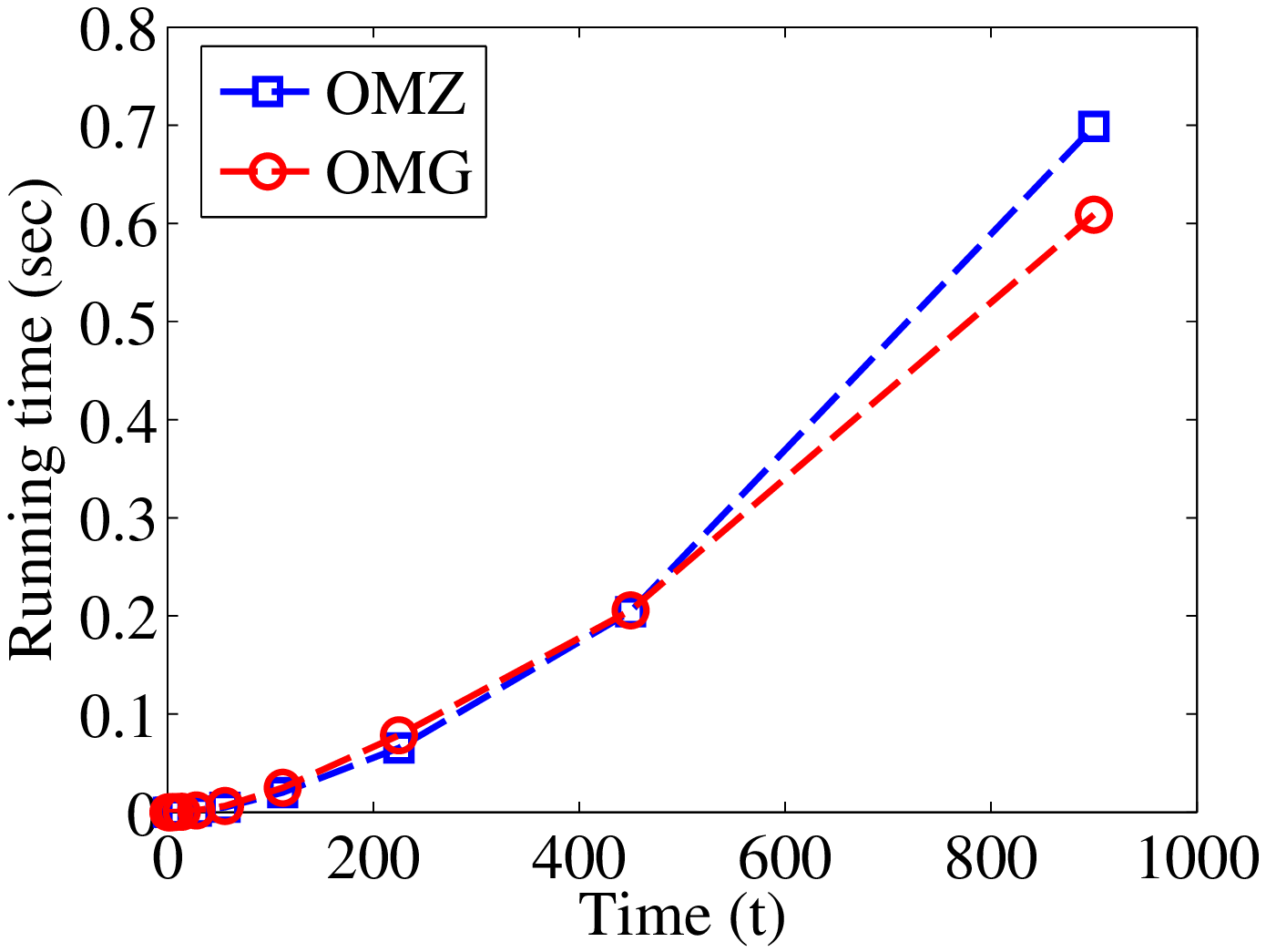}\hspace{-5pt}}
  \subfigure[Impact of $\lambda$ (at the last phase)]{
    \label{fig_runningTime2}
    \includegraphics[width=1.70in]{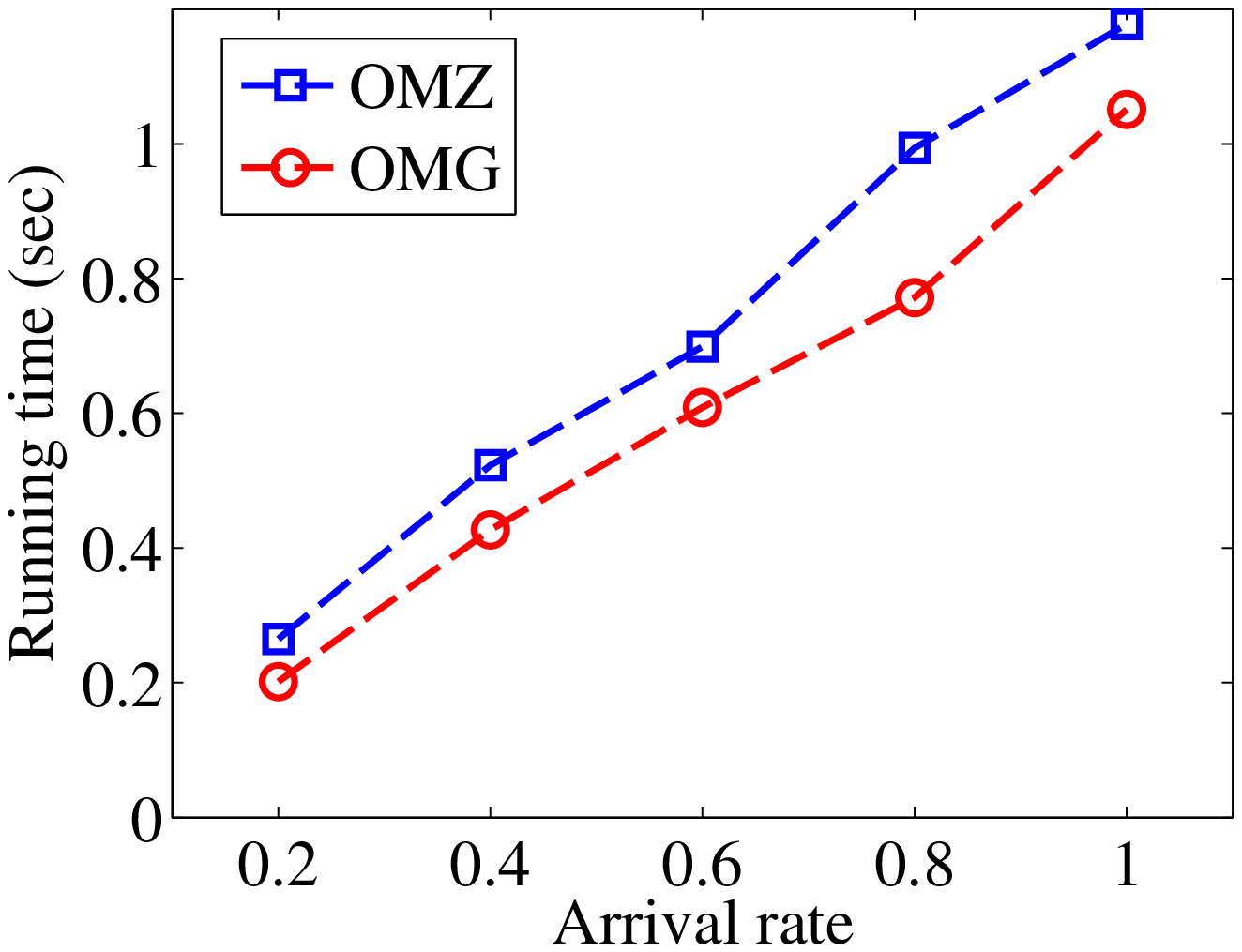}\hspace{-5pt}}
  }
  \caption{Running Time.}
  \label{fig_runningTime} 
  \vspace{-10pt}
\end{figure}
\begin{figure}[!t]
  \centering{
  \subfigure[Impact of $B$ ($\lambda=0.6$) ]{
    \label{fig_value1}
    \includegraphics[width=2.5in]{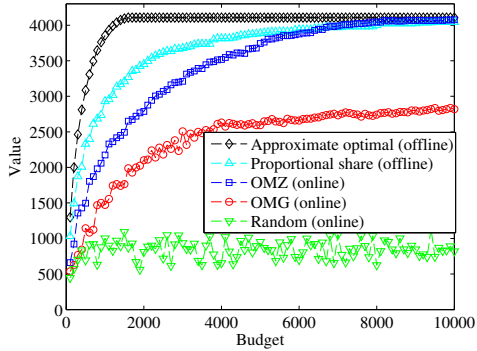}\hspace{-5pt}}
  \subfigure[Impact of $\lambda$ ($B=2000$)]{
    \label{fig_value2}
    \includegraphics[width=2.5in]{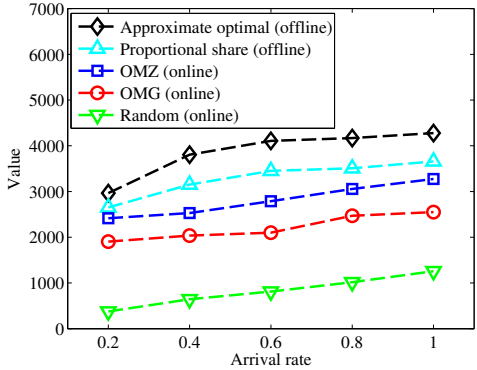}\hspace{-5pt}}
  }
  \caption{Crowdsourcer's value.}
  \label{fig_value} 
  \vspace{-10pt}
\end{figure}

\underline{Crowdsourcer's Value}:
Fig. \ref{fig_value} compares the crowdsourcer's value achieved by the \emph{OMZ} and \emph{OMG} mechanisms against the three benchmarks.
From Fig. \ref{fig_value1} we can observe that the crowdsourcer obtains higher value when the budget constraint increases.
From Fig. \ref{fig_value2} we can observe that the crowdsourcer obtains higher value when more users participate.
The \emph{approximate optimal} mechanism and the \emph{proportional share} mechanism operate in the offline scenario, where the true types or strategies of all users are known a priori, and will therefore always outperform the \emph{OMZ} and \emph{OMG} mechanisms.
It is shown that the \emph{proportional share} mechanism sacrifices some value of the crowdsourcer to achieve the cost-truthfulness compared with the \emph{approximate optimal} mechanism, and the \emph{OMG} mechanism also sacrifices some value to achieve the time-truthfulness compared with the \emph{OMZ} mechanism.
We can also observe that both the \emph{OMZ} and \emph{OMG} mechanisms are guaranteed to be within a constant factor of the offline solutions.
Specially, although the \emph{OMZ} and \emph{OMG} mechanisms are only guaranteed to be within a competitive factor of at least 8 of the \emph{proportional share} solution in expectation, as we proved in Lemma \ref{lemma:average case}, the simulation results show that this ratio is almost as small as 1.6 for \emph{OMZ} or 2.4 for \emph{OMG}.
As compared to the \emph{approximate optimal} solution, this ratio is still below 2.2 for \emph{OMZ} or below 3.4 for \emph{OMG}.
In addition, we can see that the \emph{OMZ} and \emph{OMG} mechanisms largely outweigh the \emph{random} mechanism.
\begin{figure}[!t]
  \centering{
  \subfigure[$\!\!(a_{130},\!d_{130},\!c_{130})\!\!=\!\!(269,\!269,\!3)\!$]{
    \label{fig_cost-truthful1}
    \includegraphics[width=1.70in]{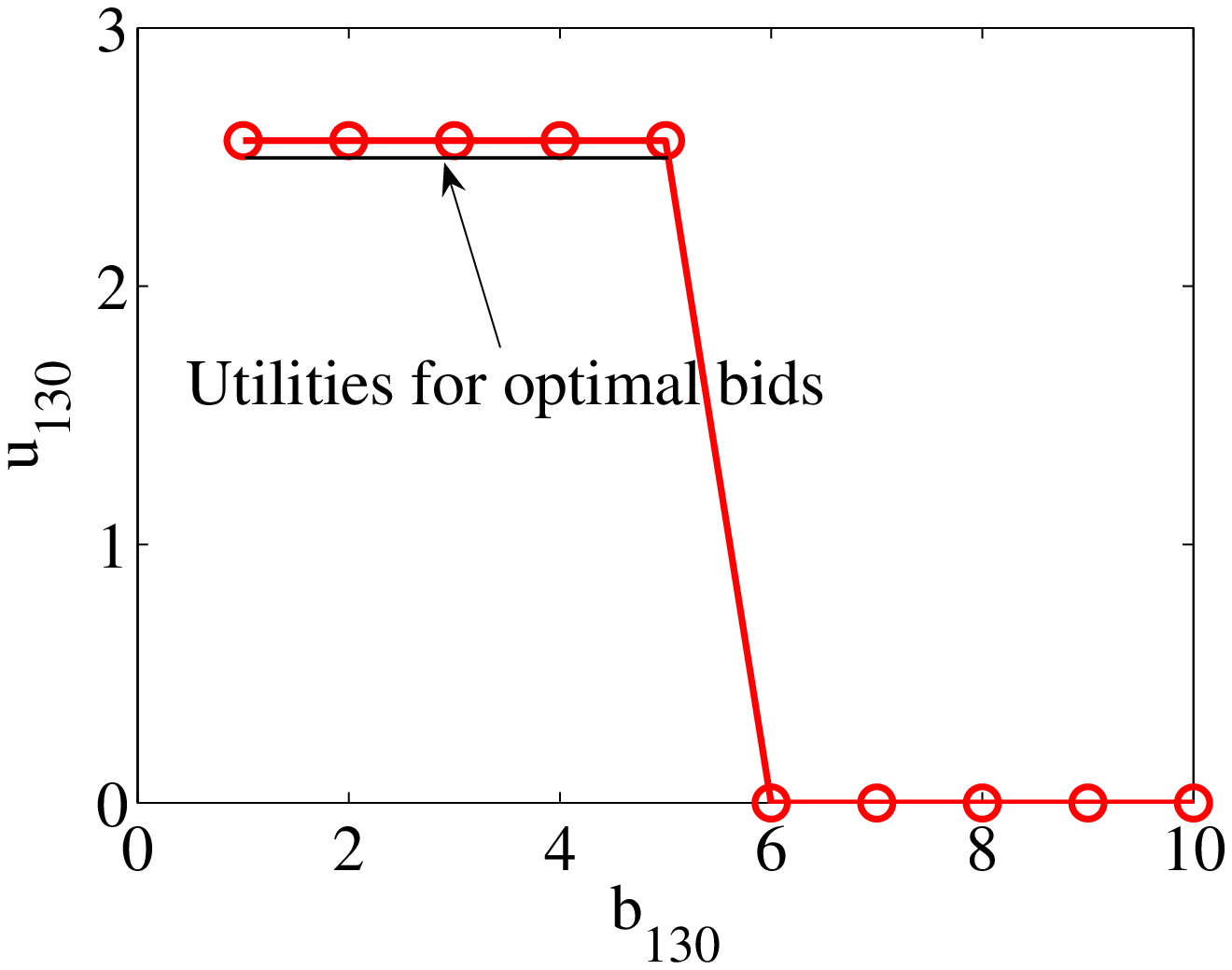}\hspace{-5pt}}
  \subfigure[$\!\!(\!a_{591},\!d_{591},\!c_{591}\!)\!\!=\!\!(\!1260,\!1260,\!8\!)\!\!$]{
    \label{fig_cost-truthful2}
    \includegraphics[width=1.70in]{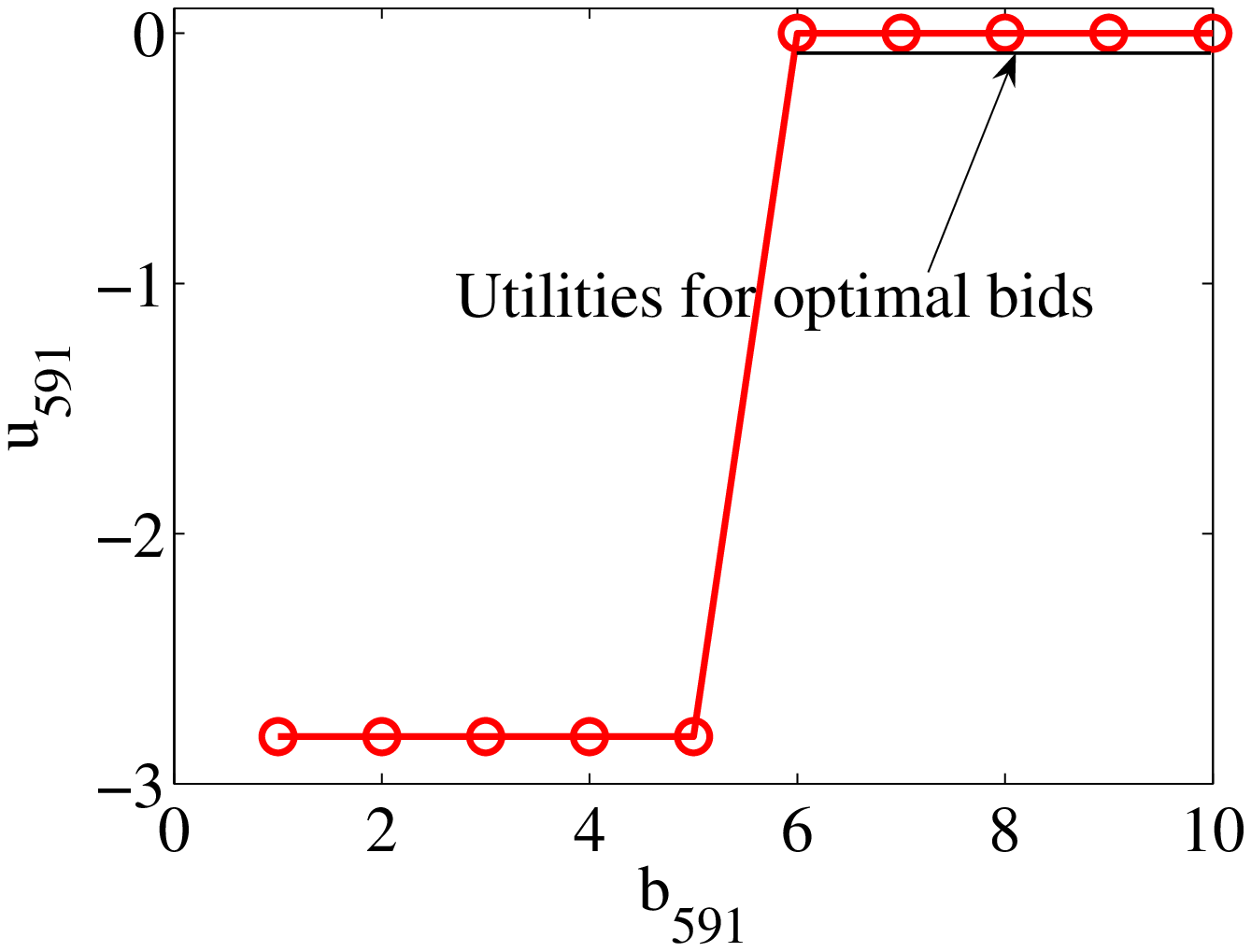}\hspace{-5pt}}
  }
  \caption{Cost-truthfulness of \emph{OMZ}.}
  \label{fig_cost-truthful} 
  \vspace{-10pt}
\end{figure}
\begin{figure}[!t]
  \centering{
  \subfigure[$(a_{17},d_{17},c_{17})\!\!=\!\!(50,123,6)$, $\hat{d}_{17}=123$, $b_{17}=6$]{
    \label{fig_time-truthful1}
    \includegraphics[width=1.70in]{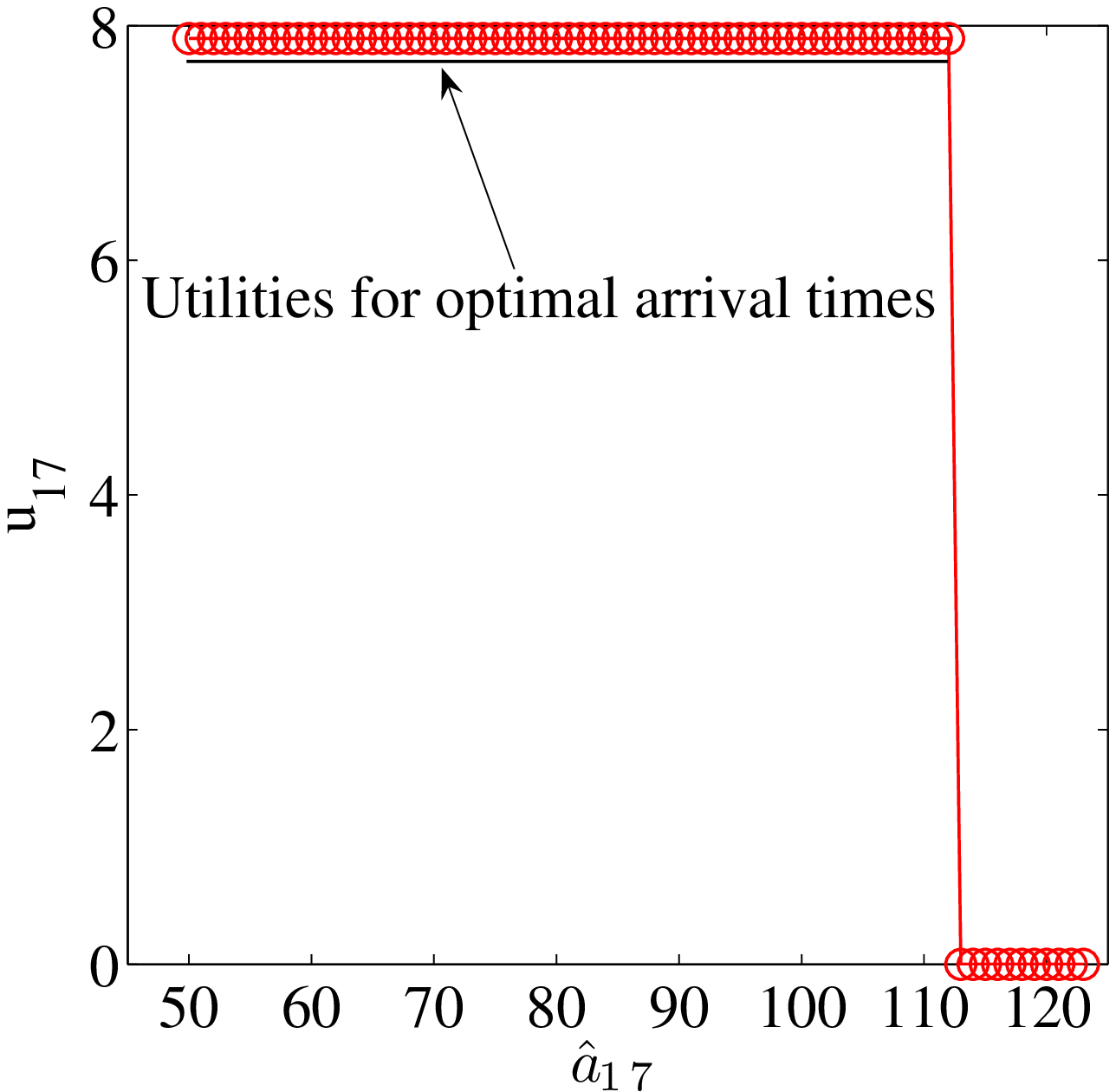}\hspace{-5pt}}
  \subfigure[$(a_{17},d_{17},c_{17})=(50,123,6)$, $\hat{a}_{17}=50$, $b_{17}=6$]{
    \label{fig_time-truthful2}
    \includegraphics[width=1.70in]{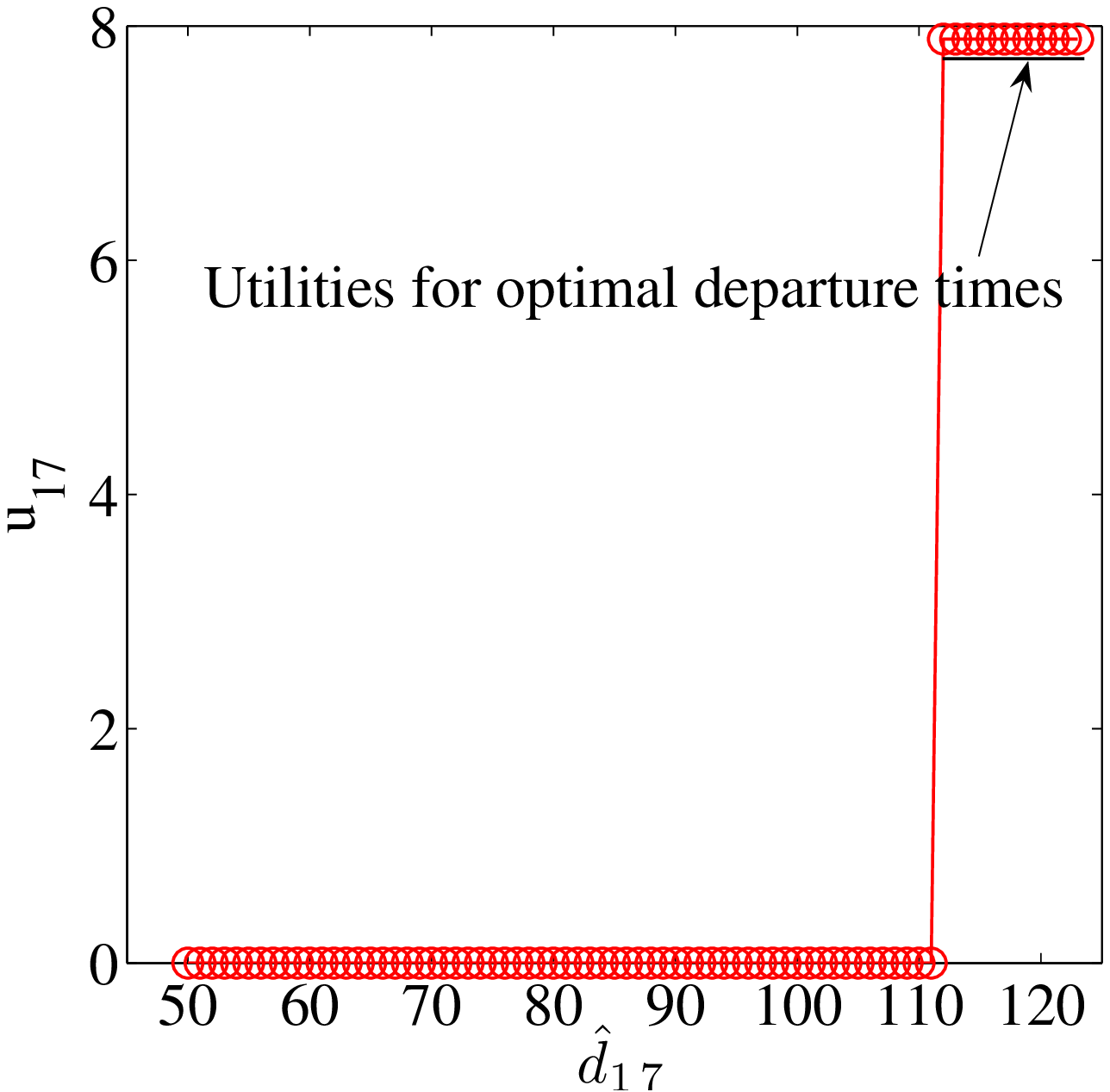}\hspace{-5pt}}
    \subfigure[$(a_{85},d_{85},c_{85})=(201,476,4)$, $\hat{d}_{85}=476$, $b_{85}=4$]{
    \label{fig_time-truthful3}
    \includegraphics[width=1.70in]{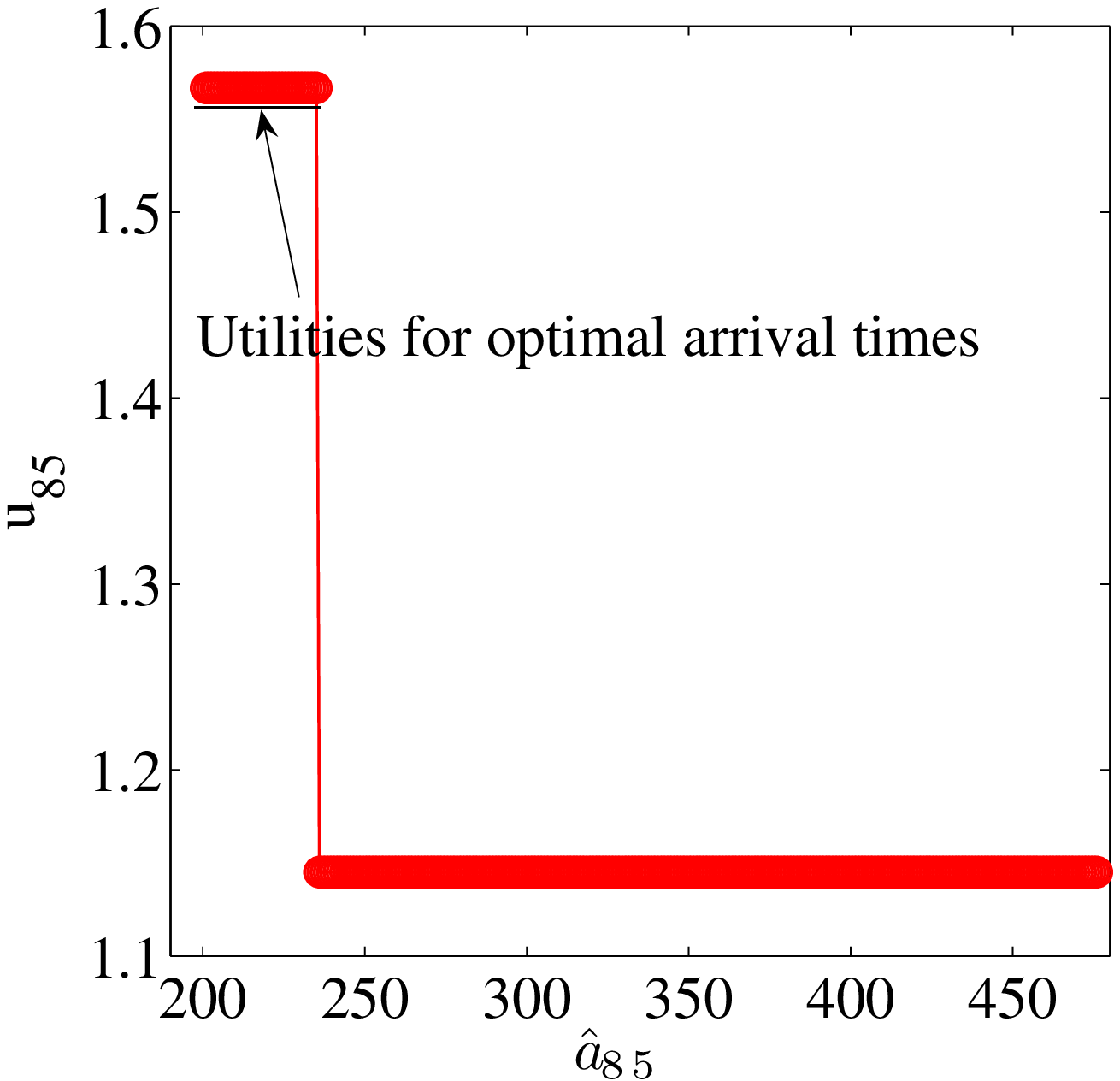}\hspace{-5pt}}
  }
  \caption{Time-truthfulness of \emph{OMG}.}
  \label{fig_time-truthful} 
  \vspace{-10pt}
\end{figure}

\underline{Truthfulness}:
We firstly verified the cost-truthfulness of \emph{OMZ} by randomly picking two users (ID=130 and ID=591) and allowing them to bid prices that are different from their true costs.
We illustrate the results in Fig. \ref{fig_cost-truthful}.
As we can see, user 130 achieves its optimal utility if it bids truthfully ($b_{130} = c_{130} = 3$) in Fig. \ref{fig_cost-truthful1} and user 591 achieves its optimal utility if it bids truthfully ($b_{591} = c_{591} = 8$) in Fig. \ref{fig_cost-truthful2}.
Then we further verified the time-truthfulness of \emph{OMG} by randomly picking two users (ID=17 and ID=85) and allowing them to report their arrival/departure times that are different from their true arrival/departure times.
We illustrate the results in Fig. \ref{fig_time-truthful}.
As shown in Fig. \ref{fig_time-truthful1} and Fig. \ref{fig_time-truthful2}, user 17 achieves its optimal utility if it reports its true arrival and departure times ($\hat{a}_{17}=a_{17}=50$, $\hat{d}_{17}=d_{123}=50$).
As shown in Fig. \ref{fig_time-truthful3}, user 85 achieves its optimal utility if it reports its true arrival time ($\hat{a}_{85}=a_{85}=201$). Note that reporting any departure time ($a_{85}\leq \hat{d}_{85} \leq d_{85}$) does not affect the utility of user 85.

\section{Related Work}
\label{sec:related work}
\subsection{Mechanism Design for Mobile Crowdsourced Sensing}
Reddy et al. \cite{reddy2010recruitment} developed recruitment frameworks to enable the crowdsourcer to identify well-suited participants for data collections.
However, they focused only on the user selection instead of the incentive mechanism design.
At present, there are only a handful of studies \cite{danezis2005much,lee2010sell,duan2012incentive,yang2012crowdsourcing,jaimes2012location} on incentive mechanism design for MCS applications in the offline scenario.
Generally, two system models were considered: the platform/crowdsourcer-centric model where the crowdsourcer provides a fixed reward to participating users, and the user-centric model where users can have their reserve prices for the sensing service.
For the crowdsourcer-centric model, incentive mechanisms were designed by using a Stackelberg game \cite{yang2012crowdsourcing,duan2012incentive}.
The Nash Equilibrium and Stackelberg Equilibrium were computed as the solution, where the costs of all users or their probability distribution was assumed to be known.
In contrast, the user-centric model can be applied to the scenario in which each user has a private cost only known to itself.
Danezis et al. \cite{danezis2005much} developed a sealed-bid second-price auction to estimate the users' value of sensing data with location privacy.
Lee and Hoh \cite{lee2010sell} designed and evaluated a reverse auction based dynamic price incentive mechanism, where users can sell their sensed data to a service provider with users' claimed bids.
Jaimes et al. \cite{jaimes2012location} proposed a recurrent reverse auction incentive mechanism with a greedy algorithm that selects a representative subset of the users according to their location given a fixed budget.
Yang et al. \cite{yang2012crowdsourcing} designed an auction-based incentive mechanism, and proved this mechanism was computationally efficient, individually rational, profitable, and truthful.
However, all of these studies failed to account for the online arrival of users.

To the best of our knowledge, there are few research work on the online mechanism design for crowdsourcing markets \cite{singer2013pricing,singla2013truthful,badanidiyuru2012learning}.
Singer et al. \cite{singer2013pricing} and Singla et al. \cite{singla2013truthful} presented pricing mechanisms for crowdsourcing markets based on the bidding model and the posted price model respectively.
However, they focused only on a simple additive utility function instead of the submodular one.
Badanidiyuru et al. \cite{badanidiyuru2012learning} considered pricing mechanisms for maximizing the submodular utility function under the bidding model.
However, they failed to consider the \emph{consumer sovereignty} and the \emph{time-truthfulness}.

\subsection{Online Auctions and Generalized Secretary Problems}
\emph{Online auction} is the essence of many networked markets, in which information about goods, agents, and outcomes is revealed one by one online in a random order, and the agents must make irrevocable decisions without knowing future information.
Our problem can be modeled as an online auction.
Combining optimal stopping theory with game theory provides us a powerful tool to model the actions of rational agents applying competing stopping rules in an online auction.

The theory of \emph{optimal stopping} is concerned with the problem of choosing a time to take a particular action, in order to maximize an expected reward or minimize an expected cost.
A classic problem of optimal stopping theory is the \emph{secretary problem}: designing an algorithm for hiring one secretary from a pool of $n$ applicants arriving online, to maximize the probability of hiring the best secretary.
Many variants of the classic secretary problem have been studied in the literature and here we review only those most relevant to this work.
An important generalization of the secretary problem is the \emph{multiple-choice secretary problem}, in which the interviewer is allowed to hire up to $k \geq 1$ applicants in order to maximize performance of the secretarial group based on their overlapping skills (or the joint utility of selected items in a more general setting).
Kleinberg \cite{kleinberg2005multiple} and Babaioff et al. \cite{babaioff2007knapsack} presented two constant competitive algorithms for a special multiple-choice secretary problem in which the objective function is a linear one, equaling to the sum of the individual values of the selected applicants.
Bateni et al. \cite{bateni2010submodular} presented a constant competitive algorithm for the \emph{submodular multiple-choice secretary problem} in which the objective function is submodular.
Another important generalization is the \emph{knapsack secretary problem}, in which each applicant also has a cost and the goal is to maximize performance of the secretarial group as along as the total cost of selected applicants does not exceed a given budget.
Babaioff et al. \cite{babaioff2007knapsack} and Bateni et al. \cite{bateni2010submodular} respectively presented constant competitive algorithms for the linear knapsack secretary problem in which the objective function is linear, and the \emph{submodular knapsack secretary problem} in which the objective function is submodular.

Our problem is similar to the \emph{submodular knapsack secretary problem} in form, but we need to consider two significant properties, the \emph{truthfulness} and the \emph{consumer sovereignty}.
Although some solutions (\cite{hajiaghayi2004adaptive,babaioff2008online,kleinberg2005multiple}) of online auctions provided good ideas of designing truthful mechanisms, they cannot be directly applied to the problem setting with submodular value function and budget constraint.
Moreover, none of these solutions considered the \emph{consumer sovereignty}.

\section{Conclusion}
\label{sec:conclusion}
In this paper, we have designed online incentive mechanisms used to motivate smartphone users to participate
in mobile crowdsourced sensing, which is a new sensing paradigm allowing us to efficiently collect data for numerous novel applications.
Compared with existing offline incentive mechanisms, we focus on a more real scenario where users arrive one by one online.
We have modeled the problem as an online auction in which the users submit their private types to the crowdsourcer over time, and the crowdsourcer aims to select a subset of users before a specified deadline for maximizing the total value of the services provided by selected users under a budget constraint.
We focus on the monotone submodular value function that can be applied in many real scenarios.
We have designed two online mechanisms under different assumptions: \emph{OMZ} can be applied to the zero arrival-departure interval case where the arrival time of each user equals to its departure time, and \emph{OMG} can be applied to the general case.
We have proved OMZ satisfies 1) computational efficiency, meaning that it can run in real time; 2) individual rationality, meaning that each participating user will have a non-negative utility; 3) budget feasibility, meaning that the crowdsourcer's budget constraint will not be violated; 4) cost-truthfulness, meaning that no user can improve its utility by reporting an untruthful cost; 5) consumer sovereignty, meaning that each participating user will have a chance to win the auction; and 6) constant competitiveness, meaning that it can perform close to the optimal solution in the offline scenario.
We have also proved OMG satisfies all the above properties as well as time-truthfulness, meaning that no user can improve its utility by reporting an untruthful arrival/departure time.

\bibliographystyle{IEEEtran}
\bibliography{myRef}

\section*{APPENDIX}
\subsection{Proof of Lemma \ref{lemma_valueFunction}}
Considering $V(\mathcal{S})=\sum_{j=1}^m{\min\{r_j,\sum_{i\in \mathcal{S}}v_{i,j}\}}$, for any $X\subseteq Y \subseteq \mathcal{U}$ and $x\in \mathcal{U}\backslash Y$ we have
\begin{align}
V(X\cup\{x\})-V(X)&=\sum_{j=1}^m{\min\{\max\{0,r_j-\sum_{i\in X}v_{i,j}\},v_{x,j}\}}\nonumber\\
&\geq \sum_{j=1}^m{\min\{\max\{0,r_j-\sum_{i\in Y}v_{i,j}\},v_{x,j}\}}\nonumber\\
&=V(Y\cup\{x\})-V(Y).\nonumber
\end{align}

Moreover, for any $X\subseteq \mathcal{U}$ and $x\in \mathcal{U}\backslash X$ we have $V(X\cup\{x\})-V(X)\geq 0$.
Therefore $V(\mathcal{S})$ is monotone submodular by Definition \ref{def:submodular}.

\subsection{Proof of Lemma \ref{lemma:average case}}
We consider two cases according to the total payment to the selected users at the last stage as follows.

\textbf{Case (a):} The total payment to the selected users at the last stage is at least $\alpha B$, $\alpha \in (0,1/2]$.
In this case, since each selected user has marginal density at least $\rho^*$, so we have that
\[V(Z_2')\geq \rho^* \alpha B = \frac{\alpha \rho_1' B}{\delta} = \frac{2\alpha V(Z_1')}{\delta}.\]

\textbf{Case (b):} The total payment to the selected users at the last stage is less than $\alpha B$, $\alpha \in (0,1/2]$.
There might be two reasons leading to that users from $Z_2$ are not selected in $Z_2'$.
The first case is when the marginal densities of some users from $Z_2$ are less than $\rho^*$, and thus we do not select them.
Even if these users are all in $Z_2$, their expected total payment is at most $B/2$.
Because of submodularity, the expected total loss due to these missed users is at most
\[\rho^*\cdot \frac{B}{2}= \frac{\rho_1' B}{2 \delta}=\frac{V(Z_1')}{\delta}.\]

The other case is when there is not enough budget to pay for some users whose marginal densities are not less than $\rho^*$.
It means that the payment for such a user (for example, user $i$) is larger than $(1/2-\alpha)B$, i.e., $V_i(\mathcal{S})/\rho^* > (1/2-\alpha)B$; otherwise adding this user to $Z_2'$ will not lead to that the total payment for $Z_2'$ exceeds the stage-budget $B/2$.
Because $\mathbb{E}[\rho_1'] \geq \rho$, we have that
\[\mathbb{E}[V_i(\mathcal{S})]\!>\! \mathbb{E}[\rho^*]\cdot (\frac{1}{2}-\alpha)B\!=\! \frac{(1-2\alpha)\mathbb{E}[\rho_1']B}{2\delta} \!\geq\!  \frac{(1-2\alpha)\rho B}{2\delta}.\]
Because the expected total payment to all users in $Z_2$ is at most $B/2$, there cannot be more than $(\frac{\delta}{1-2\alpha}-1)$ such users in $Z_2$.
Since the value of each user is at most $V(Z)/\omega$, the expected total loss due to these missed users is at most $(\frac{\delta}{1-2\alpha}-1)V(Z)/\omega$.
Therefore, we have that
\begin{align}
\mathbb{E}[V(Z_2')]&\geq \mathbb{E}[V(Z_2)]-(\frac{\delta}{1-2\alpha}-1)\frac{V(Z)}{\omega}-\frac{\mathbb{E}[V(Z_1')]}{\delta}\nonumber\\
&\geq \frac{V(Z)}{2}-(\frac{\delta}{1-2\alpha}-1)\frac{V(Z)}{\omega}-\frac{\mathbb{E}[V(Z_1')]}{\delta}\nonumber\\
&\geq [\frac{1}{2}-(\frac{\delta}{1-2\alpha}-1)\frac{1}{\omega}-\frac{1}{\delta}]\mathbb{E}[V(Z_1')].\nonumber
\end{align}

Considering both of case (a) and (b), the ratio of $\mathbb{E}[V(Z_2')]$ to $\mathbb{E}[V(Z_1')]$ will be at least $2\alpha/ \delta$, if it satisfies that
\begin{equation}
\label{eq_constraint}
\frac{1}{2}-(\frac{\delta}{1-2\alpha}-1)\frac{1}{\omega}-\frac{1}{\delta}=\frac{2\alpha}{\delta}.
\end{equation}

Therefore, for a specific parameter $\omega$, we can obtain the optimal ratio of $\mathbb{E}[V(Z_2')]$ to $\mathbb{E}[V(Z_1')]$ by solving the following optimization problem:
\begin{equation}
\textbf{Maximize } \frac{2\alpha}{\delta} \textbf{ subject to Eq. (\ref{eq_constraint}) and } \alpha \in (0,1/2].\nonumber
\end{equation}
\begin{figure}[!t]
  \centering{
  \subfigure[Optimal value of $\delta$]{
    \label{fig_delta}
    \includegraphics[width=2.5in]{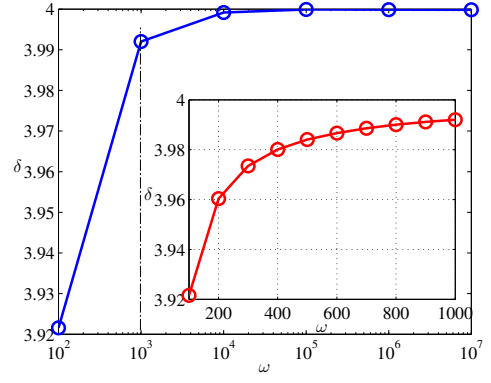}\hspace{-5pt}}
  \subfigure[Optimal ratio of $\mathbb{E}{[V(Z_2')]}$ to $\mathbb{E}{[V(Z_1')]}$]{
    \label{fig_ratio}
    \includegraphics[width=2.5in]{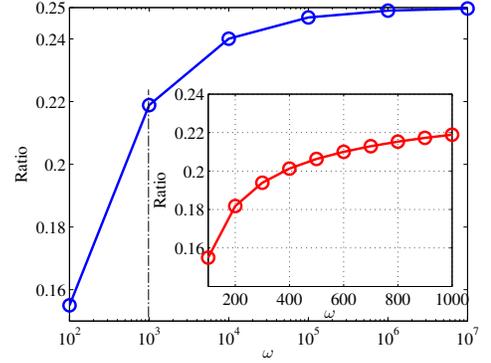}\hspace{-5pt}}
  }
  \caption{The optimal ratio of $\mathbb{E}[V(Z_2')]$ to $\mathbb{E}[V(Z_1')]$ by fixing proper $\delta$ with different values of $\omega$.}
  \label{fig_optimalRatio} 
  \vspace{-10pt}
\end{figure}

When $\omega$ is sufficiently large (at least 12), we can obtain a constant ratio of $\mathbb{E}[V(Z_2')]$ to $\mathbb{E}[V(Z_1')]$.
Fig. \ref{fig_optimalRatio} illustrates the optimal ratios that can be obtained by fixing proper $\delta$ when different values of $\omega$ are set.
As $\omega$ becomes larger, a higher ratio can be obtained.
More importantly, both the optimal ratio of $\mathbb{E}[V(Z_2')]$ to $\mathbb{E}[V(Z_1')]$ and the optimal value of $\delta$ converges fast as $\omega$ increases.
Specially, the optimal ratio approaches $1/4$ as $\omega \rightarrow \infty$ and $\delta \rightarrow 4$.

\subsection{Proof of Lemma \ref{lemma:value}}
Assume that the set of selected users computed with the budget $B'/2$ is $\mathcal{S}_l=\{1,2,\ldots,l\}$, and the set of selected users computed with the budget $B'$ is $\mathcal{S}_k=\{1,2,\ldots,k\}$.
Then, users can be sorted according to their increasing marginal densities as follows:
\begin{align}
&\frac{V_1(\mathcal{S}_0)}{b_1} \!\geq\! \frac{V_2(\mathcal{S}_1)}{b_2} \!\geq\! \cdots \!\geq\! \frac{V_l(\mathcal{S}_{l-1})}{b_l} \!\geq\! \frac{2V(\mathcal{S}_l)}{B'} \!\geq\! \frac{V_{l+1}(\mathcal{S}_l)}{b_{l+1}} \!\geq\!  \cdots \nonumber\\
&\geq \frac{V_k(\mathcal{S}_{k-1})}{b_k} \geq \frac{V(\mathcal{S}_k)}{B'} \geq \frac{V_{k+1}(\mathcal{S}_k)}{b_{k+1}} \geq \cdots \geq \frac{V_{|\mathcal{S}'|}(\mathcal{S}_{|\mathcal{S}'|-1})}{b_{|\mathcal{S}'|}}.\nonumber
\end{align}
Thus, it can be easily derived that: $V(\mathcal{S}_l)\geq V(\mathcal{S}_k)/2$.

\subsection{Proof of Lemma \ref{lemma:worst case}}
We consider two cases according to the total payment to the selected users at the last stage as follows.

\textbf{Case (a):} The total payment to the selected users at the last stage is at least $\alpha B$, $\alpha \in (0,1/2]$.
In this case, since each selected user has marginal density at least $\rho^*$, so we have that
\[V(Z_2')\geq \rho^* \alpha B = \frac{\alpha \rho_1' B}{\delta} = \frac{2\alpha V(Z_1')}{\delta}.\]

\textbf{Case (b):} The total payment to the selected users at the last stage is less than $\alpha B$, $\alpha \in (0,1/2]$.
There might be two reasons leading to that users from $Z_2$ are not selected in $Z_2'$.
The first case is when the marginal densities of some users from $Z_2$ are less than $\rho^*$, and thus we do not select them.
Even if these users are all in $Z_2$, their total payment is at most $B$.
Because of submodularity, the total loss due to these missed users is at most
\[\rho^*\cdot B= \frac{\rho_1' B}{\delta}=\frac{2V(Z_1')}{\delta}.\]

The other case is when there is not enough budget to pay for some users whose marginal densities are not less than $\rho^*$.
It means that the payment for such a user (for example, user $i$) is larger than $(1/2-\alpha)B$, i.e., $V_i(\mathcal{S})/\rho^* > (1/2-\alpha)B$; otherwise adding this user to $Z_2'$ will not lead to that the total payment for $Z_2'$ exceeds the stage-budget $B/2$.
Because $\rho_1'=2V(Z_1')/B \geq V(Z)/(4B)=\rho/4$, we have that
\[V_i(\mathcal{S})> \rho^* \cdot (\frac{1}{2}-\alpha)B= \frac{(1-2\alpha)\rho_1'B}{2\delta} \geq  \frac{(1-2\alpha)\rho B}{8\delta}.\]
Because the total payment to all users in $Z_2$ is at most $B$, there cannot be more than $(\frac{8\delta}{1-2\alpha}-1)$ such users in $Z_2$.
Since the value of each user is at most $V(Z)/\omega$, the total loss due to these missed users is at most $(\frac{8\delta}{1-2\alpha}-1)V(Z)/\omega$.
Therefore, we have that
\begin{align}
V(Z_2')&\geq V(Z_2)-(\frac{8\delta}{1-2\alpha}-1)\frac{V(Z)}{\omega}-\frac{2V(Z_1')}{\delta}\nonumber\\
&\geq \frac{V(Z)}{4}-(\frac{8\delta}{1-2\alpha}-1)\frac{V(Z)}{\omega}-\frac{2V(Z_1')}{\delta}\nonumber\\
&\geq [\frac{1}{4}-(\frac{8\delta}{1-2\alpha}-1)\frac{1}{\omega}-\frac{2}{\delta}]V(Z_1').\nonumber
\end{align}

Considering both of case (a) and (b), the ratio of $V(Z_2')$ to $V(Z_1')$ will be at least $2\alpha/ \delta$, if it satisfies that
\begin{equation}
\label{eq_constraint2}
\frac{1}{4}-(\frac{8\delta}{1-2\alpha}-1)\frac{1}{\omega}-\frac{2}{\delta}=\frac{2\alpha}{\delta}.
\end{equation}

Therefore, for a specific parameter $\omega$, we can obtain the optimal ratio of $V(Z_2')$ to $V(Z_1')$ by solving the following optimization problem:
\begin{equation}
\textbf{Maximize } \frac{2\alpha}{\delta} \textbf{ subject to Eq. (\ref{eq_constraint2}) and } \alpha \in (0,1/2].\nonumber
\end{equation}

When $\omega$ is sufficiently large, we can obtain a constant ratio of $V(Z_2')$ to $V(Z_1')$.
Specially, the optimal ratio approaches $1/12$ as $\omega \rightarrow \infty$ and $\delta \rightarrow 12$.

\subsection{Proof of Lemma \ref{lemma:truthfulness2}}
Consider a user $i$ with true type $\theta_i=(a_i,d_i,\Gamma_i,c_i)$, and reported strategy type $\hat{\theta_i}=(\hat{a_i},\hat{d_i},\Gamma_i,b_i)$.
According to the \emph{OMG} mechanism, at each time step $t\in [\hat{a_i},\hat{d_i}]$, there may be a new decision on whether to accept user $i$, and at what price.
For convenience, let $T'_t$, $B'_t$, $\rho^*_t$, and $\mathcal{S}_t$ denote the end time of the current stage, the residual budget, the current density threshold, and the set of selected users respectively at time step $t$ and before making decision on user $i$.
Let $\hat{\theta}_{-i}$ denote the strategy types of all users excluding $\hat{\theta_i}$.
We first prove the following two propositions.

\emph{Proposition (a): at some time step $t\in [\hat{a_i},\hat{d_i}]$, fix $\rho^*_t$ and $B'_t$, reporting the true cost is a dominant strategy for user $i$.}
It can be easily proved since the decision at time step $t$ is bid-independent.

\emph{Proposition (b): fix $b_i$ and $\hat{\theta}_{-i}$, reporting the true arrival/departure time is a dominant strategy for user $i$.}
It's because that user $i$ is always paid for a price equal to the maximum price attained during its reported arrival-departure interval.
Assume that user $i$ can obtain the maximum payment at time step $t\in [\hat{a_i},\hat{d_i}]$.
Then reporting an earlier arrival time or a later departure time than $t$ does not affect the payment of user $i$.
However, if user $i$ reports a later arrival time or an earlier departure time than $t$, then it will obtain a lower payment.

Based on the proposition (b), it is sufficient to prove this lemma by adding a third proposition:

\emph{Proposition (c): fix $[a_i,d_i]$ and $\hat{\theta}_{-i}$, reporting the true cost is a dominant strategy for user $i$.}
According to the proposition (a), reporting a false cost at time step $t$ cannot improve user $i$'s payment at the current time.
Thus, it only needs to prove that \emph{reporting a false cost at time step $t \in [a_i,d_i)$ still cannot improve user $i$'s payment at time step $t' (t<t'\leq d_i)$}.

Firstly, we consider the case when user $i$ is selected as a winner by reporting its true type at time step $t=a_i$.
In this case it satisfies $b_i\leq V_i(\mathcal{S}_t)/\rho^*_t \leq B'_t$, and it can obtain the payment $V_i(\mathcal{S}_t)/\rho^*_t$.
At time $t' (t<t'<T'_t)$, due to the submodularity of $V(\mathcal{S})$, we have $V_i(\mathcal{S}_{t'}) \geq V_i(\mathcal{S}_t)$.
Then user $i$ will obtain the payment $V_i(\mathcal{S}_{t'})/\rho^*_t$ if $b_i\leq V_i(\mathcal{S}_{t'})/\rho^*_t \leq B'_{t'}$, otherwise it will obtain the payment 0.
Thus, user $i$ cannot obtain higher payment at time step $t'$ than that at $t$.
It means that a user cannot improve its payment by reporting a false cost if its arrival-departure interval does not span more than one stage.

Next we consider user $i$'s payment at time step $t'(T'_t \leq t' \leq d_i)$ if its arrival-departure interval spans multiple stages.
According to the proposition (a), user $i$'s payment at time step $t'$ depends on $\rho^*_{t'}$ and $B'_{t'}$.
Because $\rho^*_{t'}$ is independent with $b_i$, it only needs to consider the effect of $b_i$ on $B'_{t'}$.
If user $i$ reports a false cost $b_i$ which still satisfies $b_i\leq V_i(\mathcal{S}_t)/\rho^*_t \leq B'_t$, then it is still accepted at price $V_i(\mathcal{S}_t)/\rho^*_t$ at time step $t$, and thus $B'_{t'}$ remains unchanged.
If user $i$ reports a larger bid $b_i>c_i$ and $b_i > V_i(\mathcal{S}_t)/\rho^*_t$, then it will not selected at time step $t$.
In this case, more budget will be allocated for other users, and $B'_{t'}$ will be diminished.
Therefore, user $i$ cannot obtain higher payment at time step $t'$.

Secondly, we consider the case when user $i$ is not selected as a winner by reporting its true type at time step $t=a_i$.
In this case it satisfies $c_i>V_i(\mathcal{S}_t)/\rho^*_t$, or $V_i(\mathcal{S}_t)/\rho^*_t > B'_t$.
In case $c_i>V_i(\mathcal{S}_t)/\rho^*_t$, if user $i$ reports a false cost $b_i$ which still satisfies $b_i>V_i(\mathcal{S}_t)/\rho^*_t$, then the outcome remains unchanged.
If user $i$ reports a lower bid $b_i<c_i$ and $b_i\leq V_i(\mathcal{S}_t)/\rho^*_t$, then it will be accepted at price $V_i(\mathcal{S}_t)/\rho^*_t$ at time step $t$.
In such case, however, its utility will be negative.
In addition, $B'_{t'}$ remains unchanged, and thus user $i$'s payment at time step $t'>t$ is not affected.
In case $V_i(\mathcal{S}_t)/\rho^*_t > B'_t$, reporting a false cost does not affect the outcome at time step $t$ or the residual budget $B'_{t'}$ at time step $t'>t$.
To sum up, reporting a false cost cannot improve user $i$'s payment at time step $t'>t$.
\end{document}